\documentclass[10pt,a4paper,reqno]{amsart}
\usepackage{amsthm,amsfonts,amssymb,amsmath,euscript,comment, enumerate,slashed,bm}
\usepackage{graphicx}
\usepackage{color}
\usepackage{bbm}
\usepackage{geometry}
\usepackage{amsfonts}
\usepackage{xcolor}
\usepackage{tikz-cd}
\usepackage{seqsplit}

\usepackage[hyperindex,breaklinks]{hyperref}

\newtheorem{thm}{Theorem}[section]

\newtheorem{df}[thm]{Definition}

\newtheorem{proposition}[thm]{Proposition}
\newtheorem{remark}[thm]{Remark}
\newtheorem{lemma}[thm]{Lemma}
\newtheorem{theorem}[thm]{Theorem}
\newtheorem{definition}[thm]{Definition}
\newtheorem{conjecture}[thm]{Conjecture}
\newtheorem{problem}[thm]{Problem}

\newcommand{\m}[1]{\mathbb{#1}}

\newcommand{\q}[1]{\mathcal{#1}}

\newcommand{\wht}[1]{\widetilde{#1}}
\newcommand{\gra}[1]{\mathbf{#1}}

\newcommand{\ep}{\varepsilon}

\newcommand{\f}{\frac}
\newcommand{\rd}{\partial}

\newcommand{\alp}{\alpha}
\newcommand{\bt}{\beta}
\newcommand{\bA}{{\bf A}}
\newcommand{\bB}{{\bf B}}

\newcommand{\mfg}{\mathfrak g}

\newcommand{\ls}{\lesssim}
\newcommand{\de}{\delta}
\newcommand{\om}{\omega}
\def\th {\theta}
\def\i {\infty}
\newcommand{\ud}{\mathrm{d}}
\newcommand{\tBox}{\widetilde{\Box}}
\newcommand{\la}{\langle}
\newcommand{\ra}{\rangle}
\newcommand{\calS}{\mathcal S}
\newcommand{\RR}{\mathbb R}
\newcommand{\nub}{\underline{\nu}}

\newcommand{\trchb}{\slashed{\mathrm{tr}}\chib}
\def\trch{\slashed{\mathrm{tr}}\chi}

\def\ub {\underline{u}}

\def\Hb {\underline{H}}
\def\chib {\underline{\chi}}
\def\chih {\hat{\chi}}
\def\chibh {\hat{\underline{\chi}}}
\def\omegab {\underline{\omega}}
\def\etab {\underline{\eta}}

\def\alp {\alpha}

\def\bt {\beta}

\def\ep {\epsilon}
\def\om {\omega}

\def\Om {\Omega}
\def\Omg{\Omega}

\numberwithin{equation}{section}

\addtocontents{toc}{\setcounter{tocdepth}{1}}

\begin{document}

\title{High-frequency solutions to the Einstein equations}

\begin{abstract}
We review recent mathematical results concerning the high-frequency solutions to the Einstein vacuum equations and the limits of these solutions. In particular, we focus on two conjectures of Burnett, which attempt to give an exact characterization of high-frequency limits of vacuum spacetimes as solutions to the Einstein--massless Vlasov system. Some open problems and future directions are discussed.
\end{abstract}

\author{C\'ecile Huneau}
\address{CMLS, Ecole Polytechnique, 91120 Palaiseau, France}
\email{cecile.huneau@polytechnique.edu}
\author{Jonathan Luk}
\address{Department of Mathematics, Stanford University, CA 94304, USA}
\email{jluk@stanford.edu}
	
	\maketitle

\tableofcontents

\section{Introduction}

In general relativity, the evolution of spacetime $(\mathcal M, g)$ is governed by the Einstein equations
\begin{equation}\label{eq:Einstein}
\mathrm{Ric}(g) - \f 12 R(g) g = 8\pi T,
\end{equation}
where $\mathrm{Ric}(g)$ and $R(g)$ are the Ricci and scalar curvature of $g$, respectively, and $T$ is the stress-energy-momentum tensor describing the matter content of the spacetime. The equations \eqref{eq:Einstein}, even in vacuum, i.e., when $T \equiv 0$, are highly nonlinear. When the gravitational radiation is sufficiently weak, the linearized Einstein equations may provide a good approximation, but in more general settings the nonlinear features of the equation play an important role. In 1968, Isaacson initiated a new perturbative scheme, inspired by the WKB analysis, to study gravitational radiation in the limit of small amplitude but high frequency (see \cite{Isaacson1, Isaacson2}). Around the same time, formal approximate solutions have been derived by Choquet-Bruhat \cite{CB.HF} in the language of geometrical optics. We refer the reader also to \cite{AliHunt, MacCallumTaub} for related constructions of approximate solutions. In these constructions, it is seen that in the high-frequency limit of vacuum solutions, a non-trivial ``effective'' stress-energy-momentum tensor is generated.

In 1989, Burnett \cite{Burnett} formulated a conjecture on the possible effects of small amplitude and high frequency perturbations in general relativity, based on the following example. This example in particular shows that high-frequency limits of vacuum solutions need not satisfy the Einstein vacuum equations. Consider a sequence of vacuum plane wave solutions
\begin{equation}
	\label{planewavl} g_\lambda=- \ud u\ud v+B_\lambda(u)^2 (e^{\omega_\lambda(u)}\,\ud x^2+e^{-\omega_\lambda(u)} \, \ud y^2),
\end{equation}
with the function $\omega_\lambda$ chosen to be of the form $\omega_\lambda(u) =\lambda\alpha(u)\cos \left(\frac{u}{\lambda}\right)$ (for some fixed function $\alp$). The Einstein vacuum equations are equivalent to the ordinary differential equation for $B_\lambda$
\begin{equation*}\label{eqB}
	B_\lambda''(u) -\omega_\lambda(u)^2 B_\lambda(u)=0.
\end{equation*}
Taking the high frequency limit, i.e., letting $\lambda \to 0$, we obtain $g_\lambda \to g_0$ uniformly, where
$$g_0=- \ud u \ud v+B_0(u)^2 (\ud x^2 + \ud y^2),$$
and
$$	B_0''(u) -\frac{1}{2}\alpha(u)^2 B_0(u)=0.$$
It can be checked that $g_0$ is no longer a solution to Einstein vacuum equations, but instead
$$\mathrm{Ric}_{\mu \nu}[g_0]=\frac{1}{2}\alpha(u)^2 \rd_\mu u \rd_\nu u.$$
The right-hand side corresponds to the stress-energy-momentum tensor of a null dust, i.e., a massless fluid without pressure; in fact, effective matter field also satisfies the propagation equation of a null dust. This can be considered as a particular case of a (measured-valued) Vlasov field.

Burnett then introduced a conjecture characterizing a general class of high-frequency limits of vacuum spacetimes. One possibility to define ``high-frequency limit'' is that in a coordinate system\footnote{The work of Burnett \cite{Burnett} only takes $K = 1$, while $K = 2$ is in the spirit of Isaacson's perturbations \cite{Isaacson1}.}
\begin{equation}\label{eq:intro.HF.def}
|\rd^k(g_i -g_0)| \ls \lambda_i^{1-k},\quad k=0,1,\cdots, K
\end{equation}
for a sequence of decreasing numbers $\{ \lambda_i\}_{i=1}^\infty$ with $\lim_{i\to \infty}\lambda_i =0$, consistent with the example considered above. Supposing that $\mathrm{Ric}(g_i) = 0$ for all $i \in \mathbb N$, we would like to understand the stress-energy-momentum tensor of the limit $g_0$. In \cite{Burnett}, Burnett made the following conjecture:
\begin{conjecture}\label{conj:forward}
A high-frequency limit of vacuum spacetimes must be isometric to a solution to the Einstein--massless Vlasov system (for a suitable Vlasov field).
\end{conjecture}
Notice that it is possible for the limit to be vacuum as well, in which the massless Vlasov field in the limit vanishes.

In the same paper \cite{Burnett}, Burnett also raised the question of whether the converse of Conjecture~\ref{conj:forward} is true. We also formulate this as a conjecture.
\begin{conjecture}\label{conj:backward}
Every (sufficiently regular) solution to the Einstein--massless Vlasov system arises as a high-frequency limit of vacuum spacetimes.
\end{conjecture}

\textbf{The purpose of this article is to survey some recent results concerning Conjectures~\ref{conj:forward} and \ref{conj:backward} as well as to describe some related open problems.}

If Conjecture~\ref{conj:forward} and Conjecture~\ref{conj:backward} are both true, then solutions to the Einstein--massless Vlasov system characterize all possible high-frequency limits. As discussed in \cite{Burnett}, physically this means that the effective matter arising from high-frequency waves propagating in different directions ``do not interact directly, but that they do affect one another by their effect on the background spacetime.''

Conjecture~\ref{conj:forward} and Conjecture~\ref{conj:backward}, though related in an obvious way, are somewhat different mathematical questions. The forward direction is mostly a question about \emph{compensated compactness}, i.e., a question concerning how weak limits behave under the special structure of the nonlinearity of the Einstein equations. In contrast, the reverse direction is a question of \emph{low-regularity existence}, since the construction of high-frequency limits necessarily involves dealing with solutions of low-regularity.

\subsection{Conjecture~\ref{conj:forward} and compensated compactness} Conjecture~\ref{conj:forward} concerns compactness, namely, it asks whether the weak limit of solutions to a partial differential equation remains a solution. For general nonlinear partial differential equations, compactness would fail since weak limits may fail to commute with nonlinearities. 

To see the role of compensated compactness, consider two sequences of functions $\{u_i\}_{i=1}^\infty$ and $\{v_i\}_{i=1}^\infty$ which converge weakly in $L^2(\mathbb R^d)$ to some limits $u_0$ and $v_0$, respectively, i.e., $\int_{\mathbb R^d} u_i w \to \int_{\mathbb R^d} u_0 w$ and $\int_{\mathbb R^d} v_i w \to \int_{\mathbb R^d} v w$ for all $L^2$ function $w$. Then in general the product $u_i v_i$ needs not converge in distribution to $u_0 v_0$, i.e., it may be that $\int_{\mathbb R^d} u_i v_i \varphi \not\to \int_{\mathbb R^d} u_0 v_0 \varphi$ for $\varphi \in C^\infty_c(\mathbb R^d)$. (Indeed, as in the plane wave example \eqref{planewavl} above, if we choose $u_i = v_i = \cos(ix)$, then\footnote{Here, and from now on, we use the notation $\rightharpoonup$ to denote weak convergence.} $u_i, v_i \rightharpoonup 0$ weakly in $L^2$, but $u_i v_i \rightharpoonup \f 12 \neq 0$.) Perhaps the most well-known example of compensated compactness is the celebrated \emph{div-curl lemma} of Murat \cite{Murat.compensation} and Tartar \cite{TartarCC}, which states that if $u_i$, $v_i$ are vector-valued i.e., $u_i,\,v_i: \mathbb R^d \to \mathbb R^d$ (with components $u_i = (u_i^1, \cdots, u_i^d)$, $v_i = (v_i^1, \cdots, v_i^d)$) such that $(u_i,v_i)$ converges to $(u_0,v_0)$ weakly in $L^2$ and that $\mathrm{div}\, u_i$ and $\mathrm{curl}\, v_i$ are bounded\footnote{More generally, one only requires that $\mathrm{div} \,u_i$ and $\mathrm{curl}\, v_i$ are compact in $H^{-1}_{\mathrm{loc}}$.} in $L^2$, then $\sum_{n=1}^d u_i^n v_i^n$ converges to $\sum_{n=1}^d u_0^n v_0^n$ in the sense of distribution.

Another example\footnote{It is possible to rephrase these examples in terms of the div-curl lemma, but in these special situations, the phenomenon can be seen more directly.} of the phenomenon of compensated compactness which is closely related to our case can be found in a system of semilinear wave equations with respect to a fixed Lorentzian metric $g$ satisfying the classical null condition, i.e., nonlinear wave equations with quadratic derivative nonlinearities consisting of \emph{classical null forms} $Q_0(\phi,\psi) = g^{\alp\bt} \rd_\alp \phi \rd_\bt\psi$ and $Q_{\alp\bt}(\phi,\psi) = \rd_\alp \phi \rd_\bt\psi - \rd_\alp \psi \rd_\bt \phi$. For concreteness, consider the semilinear model
\begin{equation}\label{eq:example.null.condition}
\Box_g \phi = Q_0(\psi,\psi),\quad \Box_g \psi = Q_{\alp\bt}(\phi,\psi).
\end{equation}
Suppose  $\{(\phi_i,\psi_i)\}_{i=1}^\infty$ is a sequence of solutions to \eqref{eq:example.null.condition} such that $(\phi_i,\psi_i) \to (\phi_0, \psi_0)$ locally uniformly and that $(\rd\phi_i, \rd \psi_i)$ are locally uniformly bounded (cf.~\eqref{eq:intro.HF.def}). Observe that as a consequence of the equations \eqref{eq:example.null.condition}, $(\Box_g \phi_i, \Box_g \psi_i)$ are also locally uniformly bounded. Therefore, after passing to a subsequence (which we do not relabel), we know that $(\rd \phi_i, \rd \psi_i, \Box_g \phi_i, \Box_g \psi_i) \rightharpoonup (\rd \phi_0, \rd \psi_0, \Box_g \phi_0, \Box_g \psi_0)$ weakly. In particular\footnote{Here, we use the standard fact that $\mbox{w-lim} (f_i h_i) = (\mbox{lim} f_i)(\mbox{w-lim} h_i)$, $\mbox{w-lim} (\rd^\bt h_i) = \rd^\bt (\mbox{w-lim} h_i)$ if $f_i$ has a uniform limit and $h_i$ has a weak limit.},
\begin{align}
	&Q_{\alpha \beta}(\phi_i,\psi_i) = \partial_\alpha \phi_i \partial_\beta \psi_i - \partial_\alpha \psi_i \partial_\beta \phi_i =\partial_\alpha (\phi_i \partial_\beta \psi_i )- \partial_\beta(\phi_i \partial_\beta \psi_i) \rightharpoonup
Q_{\alpha \beta}(\phi_0,\psi_0), \label{eq:null.CC.Qab} \\
& Q_{0}(\phi_i,\psi_i) = g^{\alpha \beta}\partial_\alpha \phi_i \partial_\beta \psi_i  = \f 12 \Box_g (\phi_i \psi_i)-\f 12 \psi_i \Box_g \phi_i - \f 12 \phi_i \Box_g \psi_i \rightharpoonup
Q_{0}(\phi_0,\psi_0). \label{eq:null.CC.Q0}
\end{align}
Consequently, as a result of the particular structure of the nonlinear terms, the limit $(\phi_0, \psi_0)$ also solves \eqref{eq:example.null.condition}.

On the other extreme, it is possible for the failure of compactness to be very severe so that a very large class of defect terms can arise in the limit. One such example is the incompressible Euler equations
\begin{equation}\label{eq:Euler}
\rd_t v + \mathrm{div}(v \otimes v) = - \nabla P,\quad \mathrm{div}\, v = 0.
\end{equation}
Given \emph{any} smooth solution to the Euler--Reynolds system
\begin{equation}\label{eq:Euler.Reynolds}
\rd_t v + \mathrm{div}(v \otimes v) = - \nabla P + \mathrm{div}\mathring{R},\quad \mathrm{div}\, v = 0
\end{equation}
with a smooth symmetric trace-free $2$-tensor $\mathring{R}$, there exists a sequence of weak solutions $\{v_i\}_{i=1}^\infty$ to \eqref{eq:Euler} such that $v_i \to v$ in $L^2$ \cite[Chapter~1]{Isett.book}. In this example, there is significant non-compactness, and weak solutions to \eqref{eq:Euler} are very flexible. 

The Einstein vacuum equations under convergence given in \eqref{eq:intro.HF.def} stand in between the two examples given above. On the one hand, compensated compactness of the type seen in the system \eqref{eq:example.null.condition} does not hold, since the weak limits of vacuum spacetimes fail to be vacuum. (This is related to the fact that the Einstein vacuum equations fail the classical null condition, for instance in generalized wave coordinates.) The assertion in Conjecture~\ref{conj:forward}, however, is that the failure of convergence would not be as flexible as in \eqref{eq:Euler}, but instead this failure of convergence satisfies a transport equation. 

This phenomenon observed in Conjecture~\ref{conj:forward} can thus be viewed as a secondary form of compensated compactness. Heuristically, high frequency limits for \emph{linear} wave equations corresponds to linear massless Vlasov equation. (This can for instance be made precise by considering the microlocal defect measures; see Section~\ref{secmes}.) In the setting of Conjecture~\ref{conj:forward}, however, the terms corresponding to the failure of convergence satisfy \emph{nonlinear} wave equations. Nevertheless, Conjecture~\ref{conj:forward} asserts that in the high frequency limit these terms still satisfy a linear massless Vlasov equation propagating on the limiting spacetime.


\subsection{Conjecture~\ref{conj:backward} and low-regularity solutions to the Einstein equations}\label{sec:intro.backward} Conjecture~\ref{conj:backward}, on the other hand, has a very different flavor. In particular, it concerns \emph{constructions} of spacetimes. Notice that any construction for Conjecture~\ref{conj:backward} is necessarily in \emph{low regularity}. Indeed, if a sequence of vacuum spacetime metrics satisfy uniform $H^s$ bounds\footnote{For the remainder of the article, $H^s$ denotes the normed space of functions with up to $s$ derivatives in $L^2$ (understood in Fourier space if $s \notin \mathbb N \cup \{0\}$, i.e., $\|f\|_{H^s} = \| (1+|\xi|)^s \mathcal F f\|_{L^2}$ for $\mathcal F$ denoting the Fourier transform).} for some $s>1$, then by Rellich's theorem, there is a subsequential limit in the $H^1$ norm. Convergence in the $H^1$ norm would then be sufficiently strong to imply that the limit is also vacuum. Thus, any construction of examples for Conjecture~\ref{conj:backward}, where the limit is not vacuum, must not have $H^s$ norms uniformly bounded. (Indeed, for high-frequency oscillations modeled upon $\lambda \alp(u) \cos(\f u \lambda)$ as in Burnett's example, the $H^s$ norm is uniformly bounded as $\lambda\to 0$ only for $s\leq 1$.) On the other hand, the best known general threshold for controlling solutions to the Einstein vacuum equations is $H^2$ \cite{L21}. Thus, any progress concerning Conjecture~\ref{conj:backward} must necessarily construct spacetimes below the regularity threshold in \cite{L21}.

There are two types of results to this regard. The first type of results concern \emph{angularly regular} spacetimes \cite{LR.HF}, where one identifies a subset of $H^1 \setminus \cup_{s>0} H^s$ for which the Einstein vacuum equations still remain well-posed. In a suitable double null coordinate system, such a subset of $H^1$ consists of metrics which are only $H^1$ along the null directions, but are more regular in the angular directions. The local well-posedness of the Einstein vacuum equations in such a low-regularity class was established in \cite{LR2} (see also \cite{LR}). Once such a well-posedness result is known, the treatment of Conjecture~\ref{conj:backward} (and in fact also Conjecture~\ref{conj:forward}) can be carried out using compactness arguments; see Section~\ref{sec:angularly.regular}

For the second type of results, one does not have a general existence result, but instead a construction is made using the special feature of the problem. This has been done in different gauges \cite{HL.HF, HL.Vlasov, Touati2}. The key point here is that even though the initial data are of very low regularity, one can carry out a specific \emph{geometric optics} type construction. In the process, one writes down the solution in a particular ansatz with the properties that (1) even though the solution with the given ansatz is only in $H^1\setminus \cup_{s>0} H^s$ for $L^2$-based Sobolev spaces, it has better integrability than is given by Sobolev embedding and is in fact Lipschitz, and (2) the term has a specific high frequency profile (see already \eqref{ansatz}), which for instance has good properties when solving elliptic equations. We note that because of the special structure of the Einstein equations, the solutions that are constructed are better behaved than general geometric optics constructions (see, e.g., \cite{Met.book, Rauch.book}). In particular, while there are nonlinear effects so as to create an effective stress-energy-momentum tensor in the limit, the propagation of the high frequency waves in different directions are effectively decoupled, and the effects of higher harmonics are very weak.

\subsection{Further related works}

We discuss some related works in addition to those that we will survey in this article. We also refer the reader to Section~\ref{sec:future} where some other related works will be mentioned.

\subsubsection{Examples of high-frequency limits}\label{sec:related.examples} Many examples of high-frequency limits that are consistent with Burnett's framework have been constructed. The earliest examples were approximate solutions given by Choquet-Bruhat \cite{CB.HF} via a geometric optics construction. She constructed metrics $\{g_\lambda\}_{\lambda \in (0,1]}$ that are almost vacuum in the sense of $\mathrm{Ric}(g_\lambda) = O(\lambda)$ and such that the limit $g_0 = \lim_{\lambda \to 0} g_\lambda$ solves the Einstein--null dust system.

For exact examples, there are explicit examples that are given in the physics literature. The first example that we are aware of is that of plane waves in \eqref{planewavl}, which was already in Burnett's paper \cite{Burnett}. A slightly more complicated, but still explicit example was given by Green--Wald in \cite{GW2}. They considered vacuum spacetimes in polarized Gowdy which can be written down explicitly, and explicitly computed their limits and the limiting stress-energy momentum tensors. In this case, the limits are solutions to the Einstein--null dust system with two families of null dust. We refer the reader to \cite{pHtF93, SGWK, SW} for more explicit examples.

Beyond explicit examples, there are also examples in plane symmetry, where solutions cannot be explicitly written down, but nonetheless there is a good existence theory for low-regularity solutions. Lott considers in his works \cite{Lott1, Lott3, Lott2} limits for polarized $\mathbb T^2$-symmetric (but non-Gowdy) spacetime. Here, the author took a slightly different perspective and considered suitable rescaled future limits of expanding cosmological spacetimes. They relied in the works \cite{LeSm, RingstromT2} on the analysis of the global solutions to the Einstein vacuum equations in order to extract limits. In subsequent works, Le Floch--Lefloch \cite{LeFLeF2019, LeFLeF2020} studied limits under more general $\mathbb T^2$ symmetry.

We remark that the examples in \cite{GW2, LeFLeF2019, LeFLeF2020, Lott1, Lott3, Lott2} all have at least a two-dimensional symmetry. Thus in principle they all fit into the framework of angularly regular spacetimes considered in Section~\ref{sec:angularly.regular} below. However, due to the exact symmetries, some of the examples considered in these works are global-in-time.

\subsubsection{Green--Wald theorem and inhomogeneities in cosmology}

In \cite{GW1}, Green--Wald proved a very general theorem concerning high-frequency limits satisfying \eqref{eq:intro.HF.def}. 

\begin{theorem}[Green--Wald \cite{GW1}]\label{thm:GW}
Suppose there is a sequence of Lorentzian metrics\footnote{To keep in line with the exposition in the rest of the article, we slightly rephrased the result in \cite{GW1} so that we consider only a sequence of metrics instead of a one-parameter family of metrics. We remark that \cite{GW1} also assumed that suitable weak limits exist along the full one-parameter family of metrics, but the proof in fact applies when weak convergence holds along a subsequence.} $\{g_i \}_{i=1}^\infty$ which satisfies the Einstein equations with matter (and possibly with a cosmological constant), i.e.,
$$\mathrm{Ric}(g_i) - \f 12 g_i R(g_i) + \Lambda g_{i} = 8\pi T_i,$$
where the stress-energy-momentum tensor $\{T_i\}_{i=1}^\infty$ are all trace free (i.e., $g_i^{\alp\bt} (T_{i})_{\alp\bt} = 0$) and satisfy the weak energy condition (i.e., $T_i(X,X) \geq 0$ for all vectors $X$ which is timelike with respect to $g_i$). Assume that $g_i$ converges to some smooth Lorentzian metric $g_0$ in the sense of condition \eqref{eq:intro.HF.def} with $K = 1$.

Then $g_0$ satisfies the Einstein equation 
$$\mathrm{Ric}(g_0) - \f 12 g_i R(g_0) + \Lambda g_{0} = 8\pi T_0,$$
where $T_0$ is trace free and satisfies the weak energy condition.
\end{theorem}

Notice that Theorem~\ref{thm:GW} in particular applies when $g_i$ are all vacuum with zero cosmological constant. If Conjecture~\ref{conj:forward} holds, then $T_0$ would correspond to the stress-energy-momentum tensor of some massless Vlasov field, which would in particular be trace free and would satisfy the weak energy condition. In other words, in the case when $g_i$ are vacuum Theorem~\ref{thm:GW} is weaker than Conjecture~\ref{conj:forward}. However, Theorem~\ref{thm:GW} is remarkable in its generality. We in particular point out the following three features, which are different from the results we will discuss later.
\begin{enumerate}
\item The theorem applies not only to a sequence of vacuum spacetimes, but instead matter fields are allowed.
\item Condition \eqref{eq:intro.HF.def} is only required for $K=1$.
\item No gauge conditions are assumed.
\end{enumerate}

The work \cite{GW1} was particularly motivated by considerations about inhomogeneities in cosmology: their result shows that within the framework of Burnett, it is impossible for inhomogeneities to mimic the effects of a positive cosmological constant. See also \cite{lotofauthors, GW2, GW.FLRW, GW.simple, Peebles}.

\subsubsection{Other forms of weak convergence in general relativity}

There are other forms of weak convergence for which one can study whether the Einstein vacuum equations are preserved. We refer the reader to \cite{fCaM2022, aMsS2023} concerning this question for a notion of metric convergence inspired by Gromov--Hausdorff type convergence. See also the work of Lott \cite{Lott1} concerning the limits for merely pointed $C^0$ convergence.   

\subsubsection{Compensated compactness in partial differential equations} The phenomenon seen in Burnett's conjectures is closely related to the theory of compensated compactness, pioneered by Tartar \cite{TartarCC} and Murat \cite{Murat.compensation}. There are many generalizations of the theory of Tartar and Murat, see \cite{mBdjC2016, mBjCDfM2009, aGbR2022, hKtY2013, Tartar1986, jwRrcRbT1987, rcRbT1988} for a sample of results.

In a different direction in the more general context of partial differential equations, compensated compactness is also useful for constructing weak solutions to nonlinear equations. We refer the reader to \cite{jmB1976, rjDP1983, Tartar.conservation.law} for examples. See also the textbook of Dafermos \cite{dafermosconstantin}.


\subsubsection{Low-regularity solutions to the Einstein equations} As mentioned in Section~\ref{sec:intro.backward} above, the construction of examples for Conjecture~\ref{conj:backward} is necessarily a low-regularity problem for the Einstein equations. In this context, the most celebrated result is the bounded $L^2$ curvature theorem of Klainerman--Rodnianski--Szeftel.
\begin{theorem}[Klainerman--Rodnianski--Szeftel \cite{L21}]\label{thm:boundedL2}
The time of existence (with respect to a maximal foliation) of a classical solution to the Einstein vacuum equations depends only on the $L^2$
norm of the curvature and a lower bound of the volume radius of the corresponding initial data set.
\end{theorem}
There was a long line of works that preceded Theorem~\ref{thm:boundedL2} regarding the low-regularity solutions to  the Einstein vacuum equations; we refer the reader to \cite{bHjyC99b,bHjyC99a,sKiR2003,sKiR2005b,sKiR2005d,SmTa}.

There is a different line of works on low-regularity solutions which are related to the angularly regular spacetimes in Section~\ref{sec:angularly.regular}. These low-regularity results were first achieved in spherical symmetry \cite{ChrSph1} and $\mathbb T^2$ symmetry \cite{pgLjS2010, LeSt}. The work of Christodoulou \cite{ChrSph1} in spherical symmetry remarkably also included the center where the symmetry degenerates. Outside exact symmetry, low-regularity results in the context of angularly regular spacetimes was motivated by impulsive gravitational waves \cite{LR,LR2} and weak null singularities \cite{LukWeakNull}. See Section~\ref{sec:angularly.regular} for a discussion of the low-regularity results. We refer the reader also to \cite{Ang20, LVdM1, LVdM2} for more recent results regarding impulsive gravitational waves.

\subsection{Outline of the paper}
The remainder of the paper is structured as follows.

We will discuss three settings for which mathematical results related to Conjectures~\ref{conj:forward} and \ref{conj:backward} have been obtained. Each of these settings involves fixing a specific gauge.
\begin{enumerate}
\item Results on Conjectures~\ref{conj:forward} and \ref{conj:backward} in an angularly regular spacetime in a double null coordinate gauge \cite{LR.HF}(\textbf{Section~\ref{sec:angularly.regular}}).
\item Results in $\mathbb U(1)$-symmetric spacetimes in an elliptic gauge. (The symmetry and gauge conditions are discussed in \textbf{Section~\ref{sec:U1.gauge}}.)
\begin{enumerate}
\item Proof of Conjecture~\ref{conj:forward} in this setting in \cite{HL.Burnett} (\textbf{Section~\ref{sec:U1.forward}}).
\item Constructions $\m U(1)$-symmetric high-frequency vacuum spacetimes and their limits in a suitable small-data regime:
\begin{enumerate}
\item The case where the limit has a finite number of families of dust \cite{HL.HF} (\textbf{Section~\ref{sec:U1.backward}}).
\item The case where the limit has a continuous Vlasov field \cite{HL.Vlasov} (\textbf{Section~\ref{vlasov}}).
\end{enumerate}
\end{enumerate}
\item Results concerning Conjectures~\ref{conj:forward} and \ref{conj:backward} in the generalized wave coordinates gauge \cite{HL.wave, Touati2} (\textbf{Section~\ref{sec:wave.coordinates}}).
\end{enumerate}

%
%
%

Finally, in \textbf{Section~\ref{sec:future}}, we will discuss some open problems and possible future directions.

\subsection*{Acknowledgements} 
We thank John Anderson, Otis Chodosh, Georgios Moschidis, Federico Pasqualotto, Igor Rodnianski and Arthur Touati for helpful discussions. This article was written for the focus issue ``The mathematics of gravitation in the non-vacuum regime,'' which arises as a follow-up of the program Mathematical perspectives of Gravitation beyond the vacuum regime at the Erwin Schr\"odinger Institute in 2022. We would like to thank the institute and the organizers for the stimulating program.

J.~Luk is partially supported by a Terman fellowship and the NSF grant DMS-2304445. 

\section{High-frequency angularly regular spacetimes}\label{sec:angularly.regular}

We begin with the setting of \emph{angularly regular spacetimes} in a double null coordinate gauge. This is a class of spacetimes for which no exact symmetries are imposed, but the spacetime metrics are more regular in some directions. There are rather complete results concerning both Conjecture~\ref{conj:forward} and Conjecture~\ref{conj:backward} in this setting. As we will see, the heart of the matter is a low-regularity local well-posedness result for this class of spacetimes, even for initial data that are in general no better than $H^1$ in terms of $L^2$-based isotropic Sobolev spaces. Because of such a general low-regularity local well-posedness result, Conjecture~\ref{conj:forward} and Conjecture~\ref{conj:backward} follow in a rather soft way, using compensated compactness type arguments.

While we emphasize that the class of angularly regular spacetimes requires no exact symmetries, it does include, a fortiori, $\mathbb T^2$ symmetric spacetimes. We refer the reader to \cite{LeFLeF2019, LeFLeF2020, LeRe, LeSm, LeSt} for a treatment directly under such a symmetry assumption.

Before we define the class of spacetimes in question, we first introduce the double null coordinate gauge.
\begin{definition}[Double null coordinates]\label{double.null.def}
A Lorentzian metric $g$ on $\mathcal M = [0,u_*] \times [0,\ub_*] \times S$ (where $S$ is a compact $2$-surface\footnote{Most of \cite{LR.HF} is stated for $S= \mathbb S^2$, but the topology of $S$ is irrelevant for the argument.} and $u_*, \ub_*>0$) is said to be in \textbf{double null coordinates} if there exists an atlas $\{U_i\}_{i=1}^N \subset S$ such that given coordinates $(\th^1, \th^2)$ in each coordinate chart $U_i$, the metric takes the form\footnote{Here, we use the convention that capital Latin letters are summed from $1$ to $2$.} 
\begin{equation}\label{double.null.coordinates}
g=-4\Omega^2 \,\ud u \ud\ub+\gamma_{AB}(\ud\th^A-b^A\ud u)(\ud\th^B-b^B\ud u),
\end{equation}
where $\Omega$ is a strictly positive function, $b = b^A\rd_{\th^A}$ is a vector field tangent to $S$ and for every $(u,\ub) \in [0,u_*]\times [0,\ub_*]$, $\gamma = \gamma_{AB}\,\ud\th^A\,\ud \th^B$ is a Riemannian metric.
\end{definition}

The gauge in Definition~\ref{double.null.def} has been useful in a variety of problems in general relativity; see for instance \cite{Chr, DHR, DHRT, DL, KN, LukWeakNull}.

As mentioned above, we are interested in the subclass of metrics taking the form \eqref{double.null.coordinates}, where the metric is still potentially of very low regularity, but is \textbf{angularly regular}, i.e., there is additional regularity in the $(\th^1,\th^2)$-directions. 

To describe this subclass of spacetimes, we consider characteristic initial data on the intersecting null hypersurfaces $H_0 \cup \Hb_0$, where $H_0 = \{0\}\times [0,\underline{I}] \times S$ and $\Hb_0 = [0,I]\times \{0\}\times S$. Denote also $S_{0,0} = H_0\cap \Hb_0 = \{0\} \times \{0\} \times S$. Consider characteristic initial data $(\gamma, \Omega,b) \restriction_{H_0\cup \Hb_0}$ where $\gamma, \Omega, b$ are the metric coefficients in \eqref{double.null.coordinates}. Assume that $b \restriction_{\Hb_0} = 0$ and that the data satisfy constraint equations and obey the following estimates (where $\rd_\vartheta$ denotes $\rd_{\th^1}$ or $\rd_{\th^2}$ derivative):
\begin{equation}\label{eq:intro.metric.bds}
\begin{split}
&\: \sum_{\mathfrak g \in \{\gamma, \log\det\gamma, \log\Om, b\}}  \sum_{i\leq 5} \|\rd_{\vartheta}^i \mathfrak g \restriction_{S_{0,0}}\|_{L^2(S)} + \sum_{i\leq 5} \|\rd_{\vartheta}^i  \rd_{\ub} b \restriction_{S_{0,0}}\|_{L^2(S)}    \\
&\: + \sum_{\mathfrak g \in \{\gamma, \log\det\gamma, \log\Om, b\}}\sum_{i\leq 5}(\|\rd_{\vartheta}^i \rd_{\ub} \mathfrak g\restriction_{H_0}\|_{L^2_{\ub} L^2(S)} + \|\rd_{\vartheta}^i \rd_u  \mathfrak g \restriction_{\Hb_0}\|_{L^2_{u} L^2(S)}) \leq C.
\end{split}
\end{equation}
We view \eqref{eq:intro.metric.bds} as a condition on \emph{angularly regular} initial data. The main local existence result for this class of data is given by the following theorem:


\begin{theorem}[L.--Rodnianski \cite{LR2}]\label{thm:LR2}
Given characteristic initial data to the Einstein vacuum equations satisfying the bounds \eqref{eq:intro.metric.bds}, there exists $\ep>0$ sufficiently small \textbf{depending only on $C$} such that for any $u_* \in (0,I]$ and $\ub_* \in (0,\ep]$, there exists a unique solution to the Einstein vacuum equations in double null coordinates in $[0,u_*]\times [0,\ub_*]\times S$ which achieves the given data. The solution is $C^\alpha\cap H^{1}$ (with $\alpha \in [0,\f 12]$) with additional regularity in $\rd_{\vartheta}$ directions, with estimates \textbf{depending only on $C$} in \eqref{eq:intro.metric.bds}. 
\end{theorem}

We will say that the solutions constructed in Theorem~\ref{thm:LR2} are \textbf{angularly regular}, since they have extra regularity in $\vartheta$.

\subsection{Compensated compactness and classification of limiting spacetimes}

Using the low-regularity local existence result in Theorem~\ref{thm:LR2}, and suitable compactness arguments, one can characterize the limit of solutions within this class. The following theorem can be viewed as a resolution of Conjecture~\ref{conj:forward} in this setting.

\begin{theorem}[L.--Rodnianski, Theorems~1.10, 4.1 in \cite{LR.HF}]\label{thm:LR3.limit}
Take a sequence of characteristic initial data to the Einstein vacuum equations which obey the bounds \eqref{eq:intro.metric.bds} uniformly. Then the following holds:
\begin{enumerate}
\item There exists a sequence of solutions $g_i$ to the Einstein vacuum equations taking the form \eqref{double.null.def} in a uniform domain of existence $[0,u_*]\times [0,\ub_*]\times S$. 
\item After passing to a subsequence (which is not relabelled), there exists a metric
$g_0$ also taking the form \eqref{double.null.coordinates} so that $g_i \to g_0$ in $C^0$ and weakly in $H^{1}$ in $[0,u_*]\times [0,\ub_*]\times S$.
\item Moreover, $g_0$ satisfies (weakly) the Einstein--null dust system with two families of null dusts. The null dusts are potentially measure-valued.
\end{enumerate}
\end{theorem}

Here, we say that $g_0$ satisfies (weakly) the Einstein--null dust system with two families of null dusts if there exists two non-negative Radon measures $\nu$ and $\underline{\nu}$ such that for any compactly supported smooth vector fields $X$, $Y$,
\begin{equation}\label{eq:double.null.weak.def}
\begin{split}
&\: \int_{[0,u_*]\times [0,\ub_*]\times S} \left((D_\mu X^\mu)(D_\nu Y^\nu)-D_\mu X^\nu D_\nu Y^\mu\right) \,\mathrm{dVol}_g \\
=&\: \int_{[0,u_*]\times [0,\ub_*]\times S} (Xu)(Yu) \,\ud\nu + \int_{[0,u_*]\times [0,\ub_*]\times S} (X\ub)(Y\ub) \,\ud\nub,
\end{split}
\end{equation}
where $u$, $\ub$ are the coordinate functions in \eqref{double.null.coordinates}. Here, the measures $\nu$ and $\nub$ correspond to the two families of null dust. As in Conjecture~\ref{conj:forward}, we allow the measures $\nu$ and/or $\underline{\nu}$ to vanish. In this case that they vanish, the limiting metric $g_0$ is vacuum.

Notice that within the class of angularly regular spacetimes, the assumptions in Theorem~\ref{thm:LR3.limit} only require uniform bounds on the sequence. This includes as a special case data satisfying the high-frequency bounds \eqref{eq:intro.HF.def}, but is much more general. In fact, Theorem~\ref{thm:LR3.limit} even allows for \emph{concentrations} (and not just oscillations) in the first derivatives of the metric.

In order to explain some ideas of the proof of Theorem~\ref{thm:LR3.limit}, we introduce in Definitions~\ref{def:null} and \ref{def:RC} some geometric constructions associated to the double null coordinates.

\begin{definition}[Null frame]\label{def:null}
The \textbf{normalized null pair} are defined as follows:
$$e_3=\Omega^{-1}(\rd_u +b^A \rd_{\th^A}),\quad e_4=\Omega^{-1}\rd_{\ub}.$$
Also, let $\{e_A\}_{A=1,2}$ denote an arbitrary local frame tangent to $S_{u,\ub}\doteq \{(u',\ub',\vartheta): u'=u,\ub'=\ub\}$.
\end{definition}

We now define the Ricci coefficients as the following $S$-tangent tensors, where $D$ is the Levi--Civita connection with respect to the spacetime metric $g$. Note that the Ricci coefficients correspond to first derivatives of the metric coefficients and are well-defined even if the metric coefficients only have $C^0\cap H^{1}_{loc}$ regularity.
\begin{definition}[Ricci coefficients]\label{def:RC}
\begin{enumerate}
\item Define the following Ricci coefficients such that for vector fields $\slashed X$, $\slashed Y$ tangential to $S$:
\begin{equation*}
\begin{split}
&\chi(\slashed X, \slashed Y)=g(D_{\slashed X} e_4, \slashed Y),\, \,\, \quad \chib(\slashed X,\slashed Y)=g(D_{\slashed X} e_3,\slashed Y),\\
&\eta(\slashed X)=-\frac 12 g(D_3 \slashed X,e_4),\quad \etab(\slashed X)=-\frac 12 g(D_4 \slashed X,e_3),\\
&\omega=-\frac 14 g(D_4 e_3,e_4),\quad\,\,\, \omegab=-\frac 14 g(D_3 e_4,e_3).
\end{split}
\end{equation*}
\item Define also 
$$\chih=\chi-\f 12 \trch \gamma, \quad \chibh=\chib-\f 12 \trchb \gamma,$$
where $\chih$ (resp.~$\chibh$) is the traceless part of $\chi$ (resp.~$\chib$) and $\trch$ (resp.~$\trchb$) is the trace of $\chi$ (resp.~$\chib$). Here, the trace taken with respect to the metric $\gamma$ on the spheres.
\end{enumerate}
\end{definition}

The key to Theorem~\ref{thm:LR3.limit} is that the estimates established in the proof of Theorem~\ref{thm:LR2} are sufficient to (i) justify that a suitable weak limit exists, and (ii) isolate the quadratic terms in the Ricci coefficients (see Definition~\ref{def:RC}) for which the product and the weak limit do not commute.
\begin{enumerate}
\item (Existence of a limit) Since the proof of Theorem~\ref{thm:LR2} gives uniform H\"older $C^\alp$ (for some $\alp \in (0,\f 12)$) and Sobolev $H^1$ estimates, standard compactness theorems then imply that the metrics has a uniform limit and that all the Ricci coefficients have at least have a weak $L^2$  a weak $L^2$ limit.
\item (Most terms do not contribute to the limiting Ricci curvature) In order to understand the limiting Ricci curvature, one needs to consider the quadratic-in-Ricci-coefficient terms (say $\Gamma_i^{(1)} \Gamma_i^{(2)}$) in the expression for the Ricci curvature and consider the defect
\begin{equation}\label{eq:angular.weak.limit.product.commute}
\mbox{w-lim} (\Gamma_i^{(1)} \Gamma_i^{(2)}) - (\mbox{w-lim\,} \Gamma_i^{(1)}) (\mbox{w-lim\,} \Gamma^{(2)}_i),
\end{equation}
where $\mbox{w-lim}$ denotes the weak limit\footnote{Here, and below, the limits only exist after passing to a subsequence. To simplify the notations, we will not relabel the subsequence.}.

It turns out that for most of these quadratic terms, either (a) one of the two factors is $\trch, \trchb, \eta,\etab$ and has a strong limit, or (b) the two factors oscillate in different directions, e.g., $\chih\cdot\chibh$ has $\chih$ oscillating in the $\ub$ direction while $\chibh$ oscillates in the $u$ direction.

In both cases (a) and (b), \eqref{eq:angular.weak.limit.product.commute} vanishes. In particular, (b) can be viewed as a form of \emph{compensated compactness}.
\item (Terms giving rise to null dust) As a result of the last point, the only quadratic terms in the Ricci coefficients that could contribute to non-trivial components of the limiting Ricci curvature are the $|\chih|_\gamma^2$ and $|\chibh|_\gamma^2$ terms in the Raychaudhuri equations:
\begin{align}
e_4 (\Omg\trch) +\frac 12 \Omg (\trch)^2=&\: -\Omg |\chih|_\gamma^2 - \Omg \mathrm{Ric}(e_4,e_4), \label{Ric44}\\
e_3 (\Omg\trchb) +\frac 12 \Omg(\trchb)^2=&\: - \Omg |\chibh|_\gamma^2 - \Omg \mathrm{Ric}(e_3,e_3). \label{Ric33}
\end{align}
The terms $\Omg |\chih|_\gamma^2$ and $\Omg |\chibh|_\gamma^2$ in \eqref{Ric44} and \eqref{Ric33} only obey uniform $L^1$ bounds (with additional angular regularity). We can therefore define the failure of convergence by the measures $\nu$ and $\underline{\nu}$ by
\begin{align*}
\nu \doteq &\: \mbox{weak-*-lim} \, (\Omg_{i}|\chih_{i}|^2_{\gamma_{i}})- (\lim \Omg_{i}) |\mbox{weak-lim} \, \chih_{i}|^2_{\gamma_0} \\ 
\underline{\nu} \doteq &\: \mbox{weak-*-lim} \, (\Omg_{i} |\chibh_{i}|^2_{\gamma_{i}} )- (\lim \Omg_{i}) |\mbox{weak-lim} \, \chibh_{i}|^2_{\gamma_0}.
\end{align*}
The measures $\nu$ and $\underline{\nu}$, which should be thought of as null dusts (see \eqref{eq:double.null.weak.def}), then correspond to $\Omg \mathrm{Ric}(e_4,e_4)$ and $\Omg \mathrm{Ric}(e_3,e_3)$ of the limiting spacetime, and these are the only non-vanishing Ricci curvature components.
\item (Propagation equation of the null dust) Finally, to complete the proof of Theorem~\ref{thm:LR3.limit}, we need to show that the two families of null dust also satisfy propagation equations. This holds because of compactness phenomena similar to those discussed in point (2) above.
\end{enumerate}

We remark that precisely because the metrics are only allowed to oscillate in the $u$ and $\ub$ directions (but not the $\th^1$, $\th^2$ directions) in the class of angularly regular spacetimes, the effective stress-energy-momentum tensor in the limit can only have two families of null dusts (corresponding to the oscillations in $u$ and $\ub$ respectively) instead of a more general massless Vlasov field.

\subsection{Construction of solutions to the Einstein--null dust system and null dust shells}

We now turn to Conjecture~\ref{conj:backward} in the setting of angularly regular spacetimes. 

\begin{theorem}[L.--Rodnianski, Theorem~4.7 in \cite{LR.HF}]\label{thm:LR3.backward}
Let $(\mathcal M =[0,u_*]\times [0,\ub_*] \times S ,\, g_0)$ be an angularly regular solution to the Einstein--null dust system with the null dusts given by angularly regular measures $\nu$ and $\underline{\nu}$. Then for any $p\in \mathcal M$, there exist $\mathcal M'\subseteq \mathcal M$ with $p\in \mathcal M'$ and a sequence of smooth angularly regular vacuum solutions $\{(\mathcal M',g_i)\}_{i=1}^\infty$ such that $g_i \to g_0$ in $C^0$ and weakly in $H^{1}$ in $\mathcal M'$.
\end{theorem}

Given Theorem~\ref{thm:LR3.limit}, in order to approximate solutions to the Einstein--null dust system by vacuum spacetimes, we only need to establish two more facts:
\begin{enumerate}
\item \emph{Initial data} to the Einstein--null dust system can be approximated by \emph{initial data} to the Einstein vacuum system (see Lemma~\ref{lem:data.approx}).
\item Uniqueness holds for the Einstein--null dust system (see Theorem~\ref{thm:uniqueness}). 
\end{enumerate} 

We will state a rough version of our data approximation lemma only on $[0,\ub_*]\times S$. (The case of $[0,u_*]\times S$ is similarly after changing $\nu \rightsquigarrow \underline{\nu}$, $\chih \rightsquigarrow \chibh$.)
\begin{lemma}[L.--Rodnianski, Proposition~9.4 in \cite{LR.HF}]\label{lem:data.approx}
Consider characteristic initial data $(g_0,\nu)$ on $[0,\ub_*]\times S$ to the Einstein--null dust system such that the following hold:
\begin{itemize}
\item The metric is $H^{1}$ is along the null direction.
\item The null dust is a non-negative Radon measure $\nu$. 
\item Both the metric and the null dust have suitable additional (Sobolev) regularity along the angular direction.
\end{itemize}
Then there exists a sequence of smooth characteristic initial data $\{g_i\}_{i=1}^\infty$ on $[0,\ub_*]\times S$ to the Einstein vacuum system such that the following hold:
\begin{itemize}
\item The sequence of metrics $g_i \to g_0$ in $C^0$ and weakly in $H^{1}$.
\item $\Omg_i|\chih_i|^2 \to \nu$ in the weak-* topology.
\item The metric satisfies the assumptions of Theorem~\ref{thm:LR2} with \textbf{uniform constants}.
\end{itemize}
\end{lemma}

Next, we state the uniqueness result. Note that uniqueness holds in general within the class of angularly regular solutions to the Einstein--null dust system. This was proven by generalizing ideas in \cite{LR2}.

\begin{theorem}[L.--Rodnianski, Theorem~4.4 in \cite{LR.HF}] \label{thm:uniqueness}
The characteristic initial value problem for the Einstein--null dust system with angularly regular initial data give rise to at most one angularly regular solution to the Einstein--null dust system.
\end{theorem}

Given Lemma~\ref{lem:data.approx} and Theorem~\ref{thm:uniqueness}, it is now easy to explain the proof of Theorem~\ref{thm:LR3.backward}, which is summarized in the following diagram.

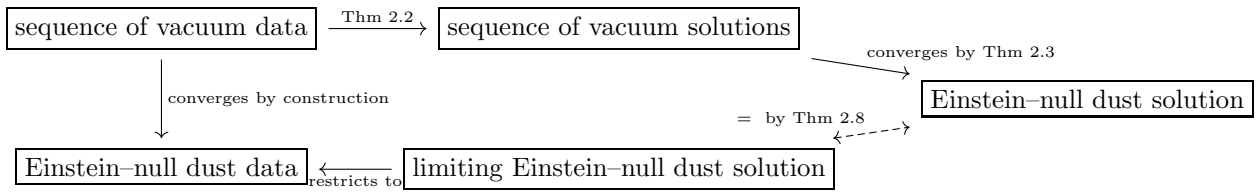
\[
\begin{tikzcd}[row sep = tiny]
\boxed{\hbox{sequence of vacuum data}} \arrow[dd, "\hbox{\tiny converges by construction}"] \arrow[r, "\hbox{\tiny Thm~\ref{thm:LR2}}"] & \boxed{\hbox{sequence of vacuum solutions}}  \arrow[dr, "\hbox{\tiny converges by Thm~\ref{thm:LR3.limit}}"] & \\
& & \boxed{\hbox{Einstein--null dust solution}}   \\
\boxed{\hbox{Einstein--null dust data}}  & \arrow[l, "\hbox{\tiny restricts to}"]  \boxed{\hbox{limiting Einstein--null dust solution}} \arrow[ur, leftrightarrow, dashed, "\hbox{\tiny = \hbox{ by Thm~\ref{thm:uniqueness}}}"] &
\end{tikzcd}
\]

Given a solution to the Einstein--null dust system, first consider its restriction to the initial data. Then approximate the initial data by vacuum data using Lemma~\ref{lem:data.approx}. Apply Theorem~\ref{thm:LR3.limit} to obtain a subsequence of vacuum solutions which converge to a solution to the Einstein--null dust system. Finally, an application of the uniqueness result in Theorem~\ref{thm:uniqueness} shows that the limiting spacetime is indeed the one we were given. This finishes the proof.

With a change of perspective, we can also use the same circle of ideas to prove existence (and uniqueness) of solutions to the Einstein--null dust system. This is given by the following theorem\footnote{There are some technical assumptions in the theorem, e.g., when $S=\mathbb S^2$, the dust needs to vanish in part of $\mathbb S^2$. We refer the reader to \cite{LR.HF} for the precise statement.}:
\begin{theorem}[L.--Rodnianski, Theorem~4.6 in \cite{LR.HF}]\label{thm:dust.local.existence}
Consider a characteristic initial value problem with the Einstein--null dust system with angularly regular initial data with a measure-valued null dust. 

Then, in an appropriate local double null domain, there exists a unique angularly regular weak solution to the Einstein--null dust system.
\end{theorem}

Here, unlike in Theorem~\ref{thm:LR3.backward}, the Einstein--null dust solution is not given, but is instead to be constructed from initial data. To do this, we approximate the given data to the Einstein--null dust system by a sequence of vacuum data, and then solve the vacuum problem and take the weak limit. The proof is represented by the following diagram:

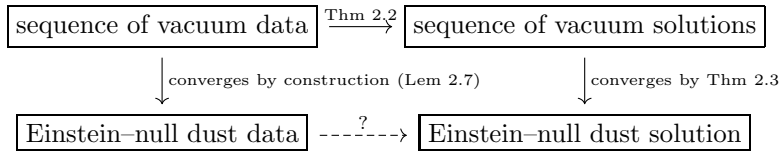
\[
\begin{tikzcd}
\boxed{\hbox{sequence of vacuum data}} \arrow[d, "\hbox{\tiny converges by construction (Lem~\ref{lem:data.approx})}"] \arrow[r, "\hbox{\tiny Thm~\ref{thm:LR2}}"] & \boxed{\hbox{sequence of vacuum solutions}}  \arrow[d, "\hbox{\tiny converges by Thm~\ref{thm:LR3.limit}}"] \\
\boxed{\hbox{Einstein--null dust data}} \arrow[r, dashed, "?"] & \boxed{\hbox{Einstein--null dust solution}}
\end{tikzcd}
\]


One particular consequence of Theorem~\ref{thm:dust.local.existence} is the construction of \emph{null dust shells} for the first time. These are solutions to the Einstein--null dust system for which the null dust is a measure which is supported on a null hypersurface. These solutions have been widely studied in the physics literature; see \cite{Barrabes.shell, BaHo, cBwIpL91, cBwIeP90, tDgtH85, tDgtH86, Hawking.shell, Penrose.shell, iR85, Synge, Tod.shell}. In fact, the construction is not restricted to null dust shell, but the null dust can be merely any measure with bounded variation and with additional angular regularity.

\section{$\mathbb U(1)$ symmetry and an elliptic gauge condition}\label{sec:U1.gauge}

From this section to Section~\ref{vlasov}, we consider Conjecture~\ref{conj:forward} and Conjecture~\ref{conj:backward} in $\m U(1)$ symmetry under an elliptic gauge condition.  In this section, we introduce the symmetry and gauge condition in Section~\ref{sec:U1} and Section~\ref{secell}, respectively. This serves as a preliminary discussion for the next few sections.

\subsection{$\m U(1)$ symmetric spacetimes}\label{sec:U1}

We fix our setup until Section~\ref{vlasov}. We will work with the  $(3+1)$-dimensional manifold ${}^{(4)}\mathcal M = \mathcal M\times \mathbb R$, where $\mathcal M = (0,T)\times \mathbb R^2$. Introduce coordinates $(t,x^1,x^2)$ on $\mathcal M$ and $(t,x^1,x^2,x^3)$ on ${}^{(4)}\mathcal M$. In the following, everything will be independent of $x^3$. We use the convention that lower case Greek indices run over $0,1,2$ while lower case Latin indices run over $1,2$. Repeated indices are summed over.

\begin{definition}[$\m U(1)$ symmetric spacetimes]
We say that Lorentzian metric ${}^{(4)}g$ on ${}^{(4)}\mathcal M$ is \textbf{$\mathbb U(1)$-symmetric} if in local coordinates, ${}^{(4)}g$ takes the form 
\begin{equation}\label{eq:4g}
	^{(4)}g= e^{-2\psi}g+ e^{2\psi}(\ud x^3+ \mathfrak A_\alp \ud x^{\alp})^2,
\end{equation}
where $g$ is a Lorentzian metric on $\mathcal M$, $\psi$ is a real-valued function on $\mathcal M$ and $\mathfrak A_{\alp}$ is a real-valued $1$-form on $\mathcal M$.

Additionally, we say that a $\mathbb U(1)$-symmetric spacetime $({}^{(4)}\mathcal M, {}^{(4)}g)$ is \textbf{polarized} if $\mathfrak A_\alp \equiv 0$.
\end{definition}

Under the $\mathbb U(1)$ symmetry assumption, the Einstein--vacuum equations verify a well-known reduction to a $(2+1)$-dimensional problem (see for instance \cite{CB.book}):
\begin{lemma}[$\m U(1)$ symmetric vacuum spacetimes]\label{lem:U1.reduction}
Suppose $({}^{(4)}\mathcal M,{}^{(4)}g)$ is a $\m U(1)$ symmetric spacetime such that the metric $^{(4)}g$ in local coordinates is given by \eqref{eq:4g}. Then the Einstein vacuum equations for $({}^{(4)}\mathcal M,{}^{(4)}g)$ reduces to the following $(2+1)$-dimensional Einstein--wave map system for $(\mathcal M, g, \psi, \varpi)$:
\begin{equation}\label{eq:U1vac}
	\left\{
	\begin{array}{l}
		\Box_g \psi + \f 12 e^{-4\psi} g^{-1}(\ud \varpi, \ud \varpi) = 0,\\
		\Box_g \varpi - 4 g^{-1} (\ud \varpi, \ud \psi) = 0,\\
		\mathrm{Ric}_{\alp\bt}(g)= 2\partial_\alp \psi \partial_\bt \psi + \f 12 e^{-4\psi} \rd_\alp \varpi \rd_\bt \varpi,
	\end{array}
	\right.
\end{equation}
where $\varpi$ is a real-valued function which relates to $\mathfrak A_\alp$ via the relation
\begin{equation}\label{eq:mfkA}
	(\ud \mathfrak A)_{\alp\bt}  = \rd_\alp \mathfrak A_\bt - \rd_\bt \mathfrak A_\alp = \f 12 e^{-4\psi} (g^{-1})^{\lambda\de} \in_{\alp\bt\lambda} \rd_\de \varpi,
\end{equation}
where $\in_{\alp\bt\lambda}$ denotes the completely antisymmetric tensor. 
\end{lemma}

\begin{remark}[The polarized subcase]\label{rmk:twist}
The function $\varpi$ is called the twist potential. In the case $({}^{(4)}\mathcal M,{}^{(4)}g)$ is polarized, $\varpi \equiv 0$, and the Einstein--wave map system \eqref{eq:U1vac} further reduces to the Einstein--(linear) scalar field system.
\end{remark}

The reduction in Lemma~\ref{lem:U1.reduction} has analogues for various Einstein--matter fields systems. For the purposes of the discussion of Burnett's conjectures in Section~\ref{sec:U1.forward}--Section~\ref{vlasov}, we define a notion of ``radially-averaged measure solutions for the restricted Einstein--massless Vlasov system in $\mathbb U(1)$ symmetry.'' This is a restricted class of $\mathbb U(1)$-symmetric solution to the Einstein--massless Vlasov system in $3+1$ dimensions where the massless Vlasov measure is additionally required to be supported in the cotangent bundle corresponding to the (2 + 1)-dimensional (instead of the (3 + 1)-dimensional) manifold.
\begin{df}[Radially-averaged measure solutions for the restricted Einstein--massless Vlasov system in $\mathbb U(1)$ symmetry]\label{def:the.final.def}
	Let $(^{(4)}\mathcal M, ^{(4)}g)$ be a $(3+1)$-dimensional $C^2$ Lorentzian manifold which is $\mathbb U(1)$-symmetric as in \eqref{eq:4g}, i.e., the metric takes the form
	$$
	^{(4)}g= e^{-2\psi}g+ e^{2\psi}(\ud x^3+ \mathfrak A_\alp \ud x^{\alp})^2,
	$$
	for $g$, $\psi$, $\mathfrak A$ independent of $x^3$. Let $\nu$ be a non-negative finite Radon measure on $S^*\mathcal M$.
	
	We say that $(^{(4)}\mathcal M, ^{(4)}g, \nu)$ is a \textbf{radially-averaged measure solution for the restricted Einstein--massless Vlasov system} in $\mathbb U(1)$ symmetry if all of the following holds:
	\begin{enumerate}
		\item The following equations are satisfied:
		\begin{equation}\label{eq:U1vac.vlasov}
			\left\{
			\begin{array}{l}
				\Box_g \psi + \f 12 e^{-4\psi} g^{-1}(\ud \varpi, \ud \varpi) = 0,\\
				\Box_g \varpi - 4 g^{-1}(\ud \varpi , \ud \psi) = 0,\\
				\int_{\mathcal M} \mathrm{Ric}(g)(Y,Y) \, \mathrm{dVol}_g=  \int_{\mathcal M} [2 (Y \psi)^2 + \f 12 e^{-4\psi} (Y\varpi)^2] \, \mathrm{dVol}_g + \int_{{S}^*\mathcal M} \langle \xi,Y\rangle^2 \, \ud \nu,
			\end{array}
			\right.
		\end{equation}
		for every $C^\infty_c$ vector field $Y$, where $\varpi$ relates to $\mathfrak A_\alp$ via \eqref{eq:mfkA}
		and where $S^*\q M$ is the cosphere bundle given by $S^*\q M= (T^* \q M\setminus\{0\})/\sim$, where $(x,\xi)\sim (y,\eta)$ if and only if $x=y$ and $\xi = \lambda \eta$ for some $\lambda>0$.
		\item The measure $\nu$ is supported on the zero mass shell, i.e., 
		\begin{equation}\label{eq:massless.def}
		(g^{-1})^{\alp\bt} \xi_\alp\xi_\bt \ud \nu \equiv 0,\quad \forall (x,\xi) \in S^* \q M.
		\end{equation}
		\item For any $C^1$ function $\widetilde{a}:T^*\mathcal M\to \mathbb R$ which is compactly supported in $x$ and such that $\widetilde{a}(x, \lambda\xi) = \lambda \widetilde{a}(x,\xi)$ $\forall \lambda >0$, 
		\begin{equation}\label{eq:transport.def}
			\int_{{S}^*\mathcal M} \{(g^{-1})^{\alp\bt} \xi_\alp\xi_\bt, \widetilde{a}(x,\xi)\} \, \ud\nu = 0,
		\end{equation}
		where $\{f(x,\xi),h(x,\xi)\} \doteq \rd_{\xi_\alp} f \rd_{x^\alp} h - \rd_{\xi_\alp} h \rd_{x^{\alp}} f$. 	\end{enumerate}
\end{df}

\begin{remark}
The condition \eqref{eq:transport.def} in Definition~\ref{def:the.final.def} can be viewed as a weak formulation of the Vlasov equation. Indeed, if the measure $\nu$ is absolutely continuous with respect to the measure on zero mass shell induced by $g$, then \eqref{eq:transport.def} is equivalent to the usual Vlasov equation. See \cite[Proposition~2.2]{HL.Burnett} for details.

\end{remark}

\begin{remark}[Null dust as a special case of massless Vlasov]\label{rmk:null.dust.as.vlasov}
Definition~\ref{def:the.final.def} includes as a special case where $\nu$ is a finite sum of delta measures, which corresponds to a finite number of families of null dust. (See \cite[Lemma~2.8]{HL.Burnett} for details.)

This will play an important role in the discussions in Sections~\ref{sec:U1.backward} and \ref{vlasov}.
\end{remark}

\subsection{Elliptic gauge condition}\label{secell}

For our discussions, there are two (related) advantages of $\m U(1)$-symmetry. First, in the vacuum case, the dynamical degrees of freedom are completely isolated into the ``wave map part'' in the functions $(\psi,\varpi)$. Second, the problem is reduced to $(2+1)$ dimensions, where the Ricci curvature tensor completely determines the Riemann curvature tensor. As a result, as long as a suitable smallness condition is imposed, one can introduce a global elliptic gauge in $(2+1)$ dimensions, which satisfy the properties in Definition~\ref{def:elliptic.gauge}. (Notice that \eqref{eq:uniformization} is related to the uniformization theorem, and relies on having two space dimensions.) As we will see below (see Lemma~\ref{lem:elliptic}), in the elliptic gauge, the metric components can be estimated elliptically from the wave map $(\psi,\varpi)$.

\begin{definition}[Elliptic gauge]\label{def:elliptic.gauge}
Suppose the $(2+1)$-dimensional metric $g$ on $\q M$ takes the form
\begin{equation}\label{g.form.0}
	g=-n^2 \ud t^2 + \bar{g}_{ij}(\ud x^i + \beta^i \ud t)(\ud x^j + \beta^j \ud t).
\end{equation}
We say that $g$ is in an \textbf{elliptic gauge} if the following holds:
\begin{enumerate}
\item $\bar{g}$ is conformally flat, i.e., there exists a function $\gamma$  such that
	\begin{equation}\label{eq:uniformization}
		\bar{g}_{ij}=e^{2\gamma}\delta_{ij},
	\end{equation}
	where $\de$ is the Euclidean metric.
\item The constant $t$-hypersurfaces $\Sigma_t\doteq\{(s,x^1,x^2): s=t\}$ are maximal, i.e., for $e_0= \partial_t -\beta^i\rd_i$ (a future-directed normal to $\Sigma_t$), $K_{ij}=-\frac{1}{2N}\q L_{e_0} \bar{g}_{ij}$ (the second fundamental form), it holds that
$$\tau \doteq \mathrm{tr}_{\bar{g}} K = 0.$$
\end{enumerate} 
\end{definition}

The gauge defined above is indeed ``elliptic'' in the sense that the metric components can be recovered from the Ricci curvature tensor via (semilinear) elliptic equations:
\begin{lemma}\label{lem:elliptic}
Under the gauge conditions in Definition~\ref{def:elliptic.gauge}, the metric components $n$, $\gamma$ and $\bt^i$ satisfy the following elliptic equations, where $H$ is defined by \eqref{elliptic.4}. (These can be derived from \cite[Appendix~B]{HL.elliptic} and some algebraic manipulations.)
\begin{align}
	&{\de^{bk}}\partial_{{k}} H_{bj}= -\frac{e^{2\gamma}}{n}\mathrm{Ric}_{0j}, \label{elliptic.1}\\
	&\Delta \gamma = -\frac{e^{2\gamma}}{n^2}G_{00} - \frac{1}{2}e^{-2\gamma}|H|^2,\label{elliptic.2}\\
	&\Delta n = ne^{-2\gamma}|H|^2 -\f 12 e^{2\gamma} n R + \f{e^{2\gamma}}{n} G_{00},\label{elliptic.3}\\
	& (\mathfrak L\beta)_{ij}=2ne^{-2\gamma}H_{ij},\label{elliptic.4}
\end{align}
where $\mathrm{Ric}_{\alp\bt}$ is the Ricci tensor, $R$ is the scalar curvature, $G_{\alp\bt} = \mathrm{Ric}_{\alp\bt} - \f 12 R g_{\alp\bt}$ is the Einstein tensor.

The Laplacian $\Delta$ and the conformal Killing operator  $\mathfrak L$ are given by 
\begin{equation}\label{L.def}
	\Delta f\doteq \de^{j\ell} \rd^2_{j\ell} f,\quad (\mathfrak LY)_{ij}\doteq \delta_{j\ell}\rd_i Y^\ell+\delta_{i\ell}\rd_j Y^\ell-\delta_{ij}\rd_k Y^k,
\end{equation}
and are both elliptic operators.

%
%
%
%
\end{lemma}

\begin{remark}
In the case of the vacuum equations, the Ricci curvature terms in Lemma~\ref{lem:elliptic} can be computed using \eqref{eq:U1vac}.
\end{remark}

\begin{remark}[Local well-posedness in elliptic gauge]\label{rmk.touati.local}
The Einstein vacuum equations in $\m U(1)$ symmetry are locally well-posed in the elliptic gauge given in Definition~\ref{def:elliptic.gauge}. For initial data which are sufficiently regular and polarized, and such that the wave part is small in $W^{1,\infty}$, this was proven in \cite{HL.elliptic}. The result was later extended by Touati \cite{Touati.local} to the general non-polarized case with smallness only imposed on the $W^{1,4}$ norm of the wave part.

The elliptic gauge in $\m U(1)$ symmetry is useful in other low-regularity problems, see for instance \cite{LVdM1, LVdM2}.
\end{remark}

\section{Burnett's conjecture in $\m U(1)$ symmetry}\label{sec:U1.forward}

In \cite{HL.Burnett}, we prove Burnett's conjecture (Conjecture~\ref{conj:forward}) when the metrics admit a $\m U(1)$ symmetry (see Section~\ref{sec:U1}) and obey the elliptic gauge condition in Section~\ref{secell}. A rough statement of our theorem is the following:
\begin{thm}[H.--L., Theorems~4.1, 4.2 in \cite{HL.Burnett}]\label{thburnett}
	Let $\{ h_i=(g_i,\psi_i,\varpi_i)\}_{i=1}^\infty$ be a sequence of solutions of \eqref{eq:U1vac} on $(0,T)\times \mathbb R^2$, and suppose the elliptic gauge condition in Definition~\ref{def:elliptic.gauge} holds for all $i \in \mathbb N$. Suppose there exists a smooth $h_0=(g_0,\psi_0,\varpi_0)$, also satisfying the same elliptic gauge condition, such that
	\begin{enumerate}
		\item \label{item:burnett.1} $h_i \rightarrow h_0$ uniformly on compact sets,
		\item \label{item:burnett.2} $\partial h_i \rightharpoonup \partial h_0$ weakly in $L^{p_0}_{loc}$ with $p_0>\frac{8}{3}$.
	\end{enumerate} 
	Then there exists a non-negative Radon measure $\nu$ on $\mathcal S^* \mathbb R^{2+1}$ such that \eqref{eq:U1vac.vlasov}  and \eqref{eq:massless.def} hold.
	
	If, in addition, we have
	\begin{enumerate}
	\setcounter{enumi}{2}
		\item \label{item:burnett.3} for all compact $K$ there exists $\lambda_i\to 0$ such that 
		\begin{equation}\label{eq:burnett.assumption}
		\sum_{k=0}^4 \lambda_i^{k-1}\|\partial^k (h_i-h_0)\|_{L^\infty(K)}\lesssim C_K,
		\end{equation}
	\end{enumerate}
	then \eqref{eq:transport.def} holds for $\nu$ defined above. In particular, $(g_0,\psi_0,\varpi_0,\nu)$ solves Einstein--massless Vlasov system in the sense of Definition~\ref{def:the.final.def}.
\end{thm}
Let us make some first comments on our theorem. 
\begin{itemize}
	\item The massless Vlasov measure $\nu$ in Theorem~\ref{thburnett} is chosen to be a suitably defined microlocal defect measure corresponding to the convergence of $(\psi_i, \varpi_i)$; see Proposition~\ref{prop:nu} and \eqref{def:dnu}. In particular, we use tools introduced by G\'erard \cite{Gerard} and Tartar \cite{Tartar}.
	\item More recent work by Guerra--Teixeira da Costa \cite{GuerraTeixeira} relaxed the assumptions in \eqref{eq:burnett.assumption}. In particular, they only imposed assumptions up to second order derivatives.
	\end{itemize}

Before giving some ideas of the proof, we recall some notions about microlocal defect measures in Section~\ref{secmes}. Then in Section~\ref{secpart1} we will sketch the proof of the first part of Theorem~\ref{thburnett}. Finally, in Section~\ref{secpart2}  we sketch the proof of the second part of Theorem~\ref{thburnett}.

\subsection{Preliminaries on microlocal defect measures}\label{secmes}


Let $\{u_i\}_{i=1}^\infty$ be a sequence of functions $\Omega \to \m R$, where $\Omega\subset \m R^k$ is open, which converges \emph{weakly} in $L^2(\Omega)$ to a function $u$. In general, after passing to a subsequence, $|u_i|^2 - |u|^2$ converges to a non-zero measure. The support of this measure is the position at which strong convergence fails. The microlocal defect measure, in contrast, is a tool which captures both the position and the frequency of this failure of strong convergence.

For instance, if $u_i=i^\frac{d}{2}\chi(i(x-x_0))$ (with $\chi \in C^\infty_c$) so that $|u_i|^2$ concentrates to a delta measure, then the corresponding microlocal defect measure is given by $\de_{x_0} \otimes \nu$, where $\de_{x_0}$ is the spatial delta measure and $\nu$ is a uniform measure on the cotangent space. On the other hand, suppose $u_i(x)=\chi(x)\cos\left(i(x\cdot\omega)\right)$ so that $u_i$ oscillates in a particular frequency $\omega$. Then the corresponding microlocal defect measure is $|\chi|^2 \ud x \otimes \de_{\omega}$, where $\de_{\omega}$ is the delta measure concentrated at the  frequency $\omega$. See \cite{Tartar} for further discussions.

Before defining microlocal defect measures, we need to recall some objects of pseudo-differential calculus. In the rest of this subsection, we fix\footnote{$k$ will be $3=2+1$ in the rest of this section, and will be $4=3+1$ later in Section~\ref{sec:wave.Burnett}.} $k \in \mathbb N$. Denote by $T^*\mathbb R^k$ the cotangent bundle of $\mathbb R^k$ with coordinates $(x,\xi)\in \mathbb R^k \times \mathbb R^k$.

\begin{df}
	\begin{enumerate}
		\item For $m\in \mathbb Z$, define the symbol class 
		$$S^m\doteq \{a:T^*\mathbb R^k \to \mathbb C: a\in C^{\infty},\, \forall \alp,\bt,\, \exists C_{\alp,\bt}>0,\, |\rd_x^\alp \rd_\xi^\bt a(x,\xi)| \leq C_{\alp,\bt} (1+|\xi|)^{m-|\bt|} \}.$$
		\item Given a symbol $a \in S^m$, define the operator $\mathrm{Op}(a):\mathcal S(\mathbb R^k)\to \mathcal S(\mathbb R^k)$ by
		$$(\mathrm{Op}(a)u)(x) \doteq \f{1}{(2\pi)^k}\int_{\mathbb R^k} \int_{\mathbb R^k} e^{i(x-y)\cdot \xi} a(x,\xi) u(y) \,\ud y\, \ud \xi.$$
		We say that $A=\mathrm{Op}(a)$ is a pseudo-differential operator of order $m$ with symbol $a$. If moreover {$a(x,\xi) = a_{\mathrm{prin}}(x,\xi) + a_{\mathrm{error}}(x,\xi)$ for $|\xi|\geq 1$, where $a_{\mathrm{prin}}(x,\lambda\xi) = \lambda^m a(x,\xi)$ for all $\lambda>0$, and $a_{\mathrm{error}} \in S^{m-1}$,} we say that ${a_{\mathrm{prin}}}$ is the principal symbol of $A$.
	\end{enumerate}
\end{df}

We now turn to the discussion of microlocal defect measures \cite{Gerard, Tartar}. The following is a special case of \cite[Theorem~1]{Gerard}. 

%
%

\begin{thm}[Existence of microlocal defect measures]\label{thm:existenceMDM}
	Let $\{u_i\}_{i=1}^{\infty} \in L^2(\mathbb R^k;\mathbb C)$ be a bounded sequence in $L^2(\mathbb R^k;\mathbb C)$ such that $u_i \rightharpoonup 0$ weakly in $L^2(\mathbb R^k;\mathbb C)$.
	
	Then there exists a subsequence $\{u_{i_k}\}_{k=1}^{\infty}$ and a non-negative Radon measure $\mu$ on $S^*\mathbb R^k$ such that the following holds for every $0$-th pseudo-differential operator $A$ of order $0$ with principal symbol $a(x,\xi)$ which is compactly supported in $x$ and satisfies $a(x,\lambda \xi) = a(x,\xi)$, $\forall \lambda >0$:
	\begin{equation}\label{eq:mdm.def}
	\lim_{k\to +\infty}  \langle A u_{i_k}, u_{i_k} \rangle_{L^2(\mathbb R^k;\mathbb C)} = \int_{S^*\mathbb R^k} a(x,\xi) \,\ud \mu.
	\end{equation}
\end{thm}

The measure $\ud \mu$ in Theorem~\ref{thm:existenceMDM} is called a \emph{microlocal defect measure} (or \emph{$H$-measure}). 

Let us already point out two important properties of microlocal defect measures that are important for the discussion below. First, they satisfy the following localization property:
\begin{thm}[Localization of microlocal defect measures, Corollary~2.2 in \cite{Gerard}]\label{thm:localization}
	Let $\{u_i\}$ be a sequence such that $u_i \rightharpoonup 0$ weakly in $L^2(\mathbb R^k,\mathbb C)$, and $\mu$ is a microlocal defect measure such that \eqref{thm:existenceMDM} holds (without passing to a subsequence). Let $P$ be an $m$-th order differential operator with principal symbol $p = \sum_{|\alp|= m} a_\alp (i\xi)^\alp$ for some smooth functions $a_\alp$. If $\{Pu_i\}_{i=1}$ is relatively compact in $H^{-m}_{loc}(\mathbb R^k,\mathbb C)$, then 
	$$p\,\ud\mu = 0.$$
\end{thm}
In particular, as a consequence of Theorem~\ref{thm:localization}, it can be shown that microlocal defect measures associated to sequences of solutions to the linear wave equation are supported on the zero mass shell. The second important property is that microlocal defect measures associated to sequences of solutions to the linear wave equation also satisfy the Vlasov equation. 

\subsection{The form of the effective stress-energy tensor}\label{secpart1}
We now sketch the proof of the first part of Theorem \ref{thburnett}. Let $h_i=(g_i,\psi_i,\varpi_i)$ be a sequence of solutions of \eqref{eq:U1vac} in elliptic gauge which satisfies the conditions (\ref{item:burnett.1}) and (\ref{item:burnett.2}) of the theorem. After multiplying by suitable cut-off functions (which we suppress in the notations), we reduce the problem to compact sets and thus
$$\partial(\psi_i-\psi_0) \rightharpoonup 0, \quad \partial(\varpi_i-\varpi_0) \rightharpoonup 0 \quad \hbox{weakly in $L^2( \m R^{2+1})$.}$$

In particular, we can apply Theorem~\ref{thm:existenceMDM} to show the existence of a microlocal defect measures corresponding to the derivatives of $\psi_i - \psi_0$ and $\varpi_i - \varpi_0$. In fact, an application of Theorem~\ref{thm:localization} shows that the microlocal defect measures take the form as in the following proposition (see, e.g., \cite[Lemma~3.10]{Tartar}):
\begin{proposition}[Existence of microlocal defect measures]\label{prop:nu}
	There exist non-negative Radon measures\footnote{We note that we have used slightly different conventions here as our original \cite{HL.Burnett}: here $\ud \nu^{\psi}$, $\ud \nu^{\varpi}$ are defined to act on functions which are homogeneous of order $2$ in $\xi$, while in \cite{HL.Burnett}, they act on homogenous functions of order $0$.} $\ud \nu^\psi$, $\ud \nu^\varpi$ on $S^*\mathbb R^{2+1}$ such that after passing to a subsequence (which is again labelled as $(\psi_i, \varpi_i)$), the following holds for any $A= \mathrm{Op}(a)$, where $a(x,\xi)$ which is compactly supported in $x$ and $a(x,\lambda \xi) = a(x,\xi)$:
	\begin{align*}
	\lim_{i\to \infty} \langle \rd_\alp (\psi_i-\psi_0), A \rd_\bt(\psi_i-\psi_0) \rangle_{L^2(\mathbb R^{2+1}, g_0)} &= \int_{S^*\mathbb R^{2+1}} a \xi_\alp \xi_\bt \,\ud \nu^\psi,\\
	\lim_{i\to \infty} \langle \rd_\alp (\varpi_i-\varpi_0), A \rd_\bt(\varpi_i-\varpi_0) \rangle_{L^2(\mathbb R^{2+1}, g_0)}& = \int_{S^*\mathbb R^{2+1}} a \xi_\alp \xi_\bt \,\ud \nu^\varpi.
\end{align*}
\end{proposition}

We are now ready to pass to the weak limit in \eqref{eq:U1vac}. Recall that we have a sequence $(\psi_i, \varpi_i)$ solving
\begin{equation}\label{eq:psi.om.nonlinear}
		\left\{
\begin{array}{l}
	\Box_{g_i} \psi_i + \f 12 e^{-4\psi_i} g^{-1}(\ud \varpi_i, \ud \varpi_i) = 0,\\
	\Box_{g_i} \varpi_i - 4 g_i^{-1} (\ud \varpi_i, \ud \psi_i) = 0.
\end{array}
\right.
\end{equation}
For $\varphi \in \{\psi,\varpi\}$, $\Box_{g_i} \varphi_i = (-\det g_i)^{-\f 12} \rd_\alp (g_i^{\alp\bt} (-\det g_i)^{\f 12}\rd_\bt \varphi_i)$. Since $(-\det g_i)^{-\f 12}$ has a strong limit and $\rd_\alp (g_i^{\alp\bt} (-\det g_i)^{\f 12}\rd_\bt \varphi_i)$ is in divergence form, it follows that $\Box_{g_i} \varphi_i \rightharpoonup   \Box_{g_0} \varphi_0$ weakly. Moreover, the quadratic semilinear terms satisfies the null condition, and thus passes to the weak limit using an argument similar to \eqref{eq:null.CC.Q0}. Therefore, $(\psi_0,\varpi_0)$ satisfy the same system of equations.

Let us now consider the third equation in \eqref{eq:U1vac}:
\begin{equation}\label{ric}\mathrm{Ric}_{\alp\bt}(g_i)= 2\partial_\alp \psi_i \partial_\bt \psi_i + \f 12 e^{-4\psi_i} \rd_\alp \varpi_i \rd_\bt \varpi_i.
\end{equation}
We claim that the left-hand side converges weakly to $\mathrm{Ric}_{\alp\bt}(g_0)$.  To see this, we refer the reader to the expression of the Ricci tensor in \eqref{elliptic.1}--\eqref{elliptic.4} and
\begin{align}
	\label{Rij} \mathrm{Ric}_{jb} = &\delta_{jb}\left(-\Delta \gamma-\frac{1}{2N}\Delta n\right)
	-\frac{1}{n}(\partial_t-\beta^k\partial_k)H_{jb}-2e^{-2\gamma} H_b{ }^{\ell} H_{j\ell} \\
	\notag &+\frac{2}{n}\partial_{(j|} \beta^k H_{k|b)}-\frac{1}{n}\left( \partial_b \partial_j n -\frac{1}{2}\delta_{jb}\Delta n -\left(2\delta_{(b}^k\partial_{j)} \gamma -\delta_{jb} \de^{\ell k}\partial_{\ell} \gamma\right)\partial_k n \right).
\end{align}
The key is to note that any quadratic term in the derivatives of the metric components must either have a factor of $H$ or a factor of $\nabla \gamma$. Now even though we a priori only assume boundedness of all first derivatives, both $H_i$ and $\nabla \gamma_i$ in fact converge \emph{strongly} in $L^2$, which therefore results in $\mathrm{Ric}_{\alp\bt}(g_i) \rightharpoonup \mathrm{Ric}_{\alp\bt}(g_0)$ weakly. To see the strong convergence of $H_i$ and $\nabla \gamma_i$, we use (1) the spatial elliptic equations and (2) the equations involving $\partial_t$ that they satisfy to show that $H_i,\nabla_i \gamma$ has uniform $W^{1,\frac{p_0}{2}}(\m R^{2+1})$ bounds. By compactness of the embedding $W^{1,\frac{p_0}{2}}(\m R^{2+1}) \hookrightarrow L^2(\m R^{2+1})$, the desired strong convergence follows.


On the other hand, the weak limit of the right-hand side of \eqref{ric} can be expressed in term of the microlocal defect measures given in Proposition~\ref{prop:nu}.  Therefore, defining
\begin{equation}\label{def:dnu}
	\ud\nu\doteq 2\ud \nu^\psi + \f 12 e^{-4\psi_0} \, \ud \nu^\varpi,
\end{equation}
it follows that the last equation in \eqref{eq:U1vac.vlasov} is satisfied for every vector field $Y\in C^\infty_c(\Omega)$.

Finally, the fact that the measure $\ud \nu$ is supported on the set $\{(x,\xi) \in S^*\mathcal M: g^{-1}_0(\xi,\xi) = 0\}$ can easily be obtained from Theorem \ref{thm:localization} and the fact that $\Box_{g_0}(\psi_i-\psi_0)$ and $\Box_{g_0}(\varpi_i-\varpi_0)$ are compact in $H^{-1}$.

\subsection{The transport equation for the microlocal defect measure}\label{secpart2}

The main difficulty for Theorem~\ref{thburnett} is therefore to establish the transport equation \eqref{eq:transport.def}. As mentioned in Section~\ref{secmes}, microlocal defect measures arising from solutions to linear equations satisfy the massless Vlasov equation (see for instance \cite{Francfort, FrancfortMurat, Tartar}). We review this fact in Section~\ref{sec:Mink.warmup}, and then explain in Sections~\ref{sec:Q0.trilinear}--\ref{sec:wave.map.structure} why the transport equation still holds in our nonlinear setting.

\subsubsection{A Minkowskian warm-up}\label{sec:Mink.warmup}
We start with the simplest possible case. Denote the Minkowski metric by $m$ and suppose the linear wave equation $\Box_m \phi_i = f_i$ holds with $\partial \phi_i \rightharpoonup 0$ weakly in $L^2$ and $f_i \to 0$ in $L^2$ norm. We argue as in Proposition~\ref{prop:nu} that (after passing to a subsequence which is not relabelled) there exists a non-negative Radon measure $\nu$ such that the following holds for $A$, $a$ as in Proposition~\ref{prop:nu}:
$$\lim_{i\to \infty} \langle \rd_\alp \phi_i, A \rd_\bt\phi_i \rangle_{L^2(\mathbb R^{2+1}, m)} = \int_{S^*\mathbb R^{2+1}} a \xi_\alp \xi_\bt \,\ud \nu.$$
Moreover, arguing with Theorem \ref{thm:localization} as in Section~\ref{secpart1}, $\nu$ is supported on the zero mass shell.

We now derive an analogue of the transport equation \eqref{eq:transport.def} adapted to this setting. For a pseudo-differential operator $A$ of order $0$ with real principal symbol $a$, define
$$\mathbb T^A_{\mu\nu}[\phi_i, \phi_i] \doteq \rd_{(\mu|} \phi_i \rd_{|\nu)} A \phi_i +  \f 12 m_{\mu\nu} m^{\alp\bt} \rd_{(\alp|} \phi_i \rd_{|\bt)} A \phi_i. $$
It then follows that
\begin{equation}\label{eq:for.Mink.Stokes}
\nabla^\mu (\mathbb T^A_{\mu\nu}[\phi_i, \phi_i] (\rd_t)^\nu) = \f 12 \Box_m \phi_i \rd_t A\phi_i + \f 12 \rd_t \phi_i A \Box_m \phi_i + \f 12 \rd_t\phi_i [\Box_m, A] \phi_i,
\end{equation}
We integrate \eqref{eq:for.Mink.Stokes} and pass to the $i\to \infty$. Note that the left-hand side of \eqref{eq:for.Mink.Stokes}, being an exact divergence, is integrated away, while the first two terms on the right-hand side of \eqref{eq:for.Mink.Stokes} vanish in the $i\to \infty$ limit since $f_i \to 0$ in the $L^2$ norm. Thus, the only contribution comes from the third term on the right-hand side of \eqref{eq:for.Mink.Stokes}. Noting that the principal symbol of $[\Box_m, A]$ is proportional to $\{m^{\alp\bt}\xi_\alp\xi_\bt, a\}$ (see, e.g., \cite[Theorem 2, p.237]{Stein}), we obtain
\begin{equation}\label{eq:stupid.model.almost}
\xi_t \{m^{\alp\bt}\xi_\alp\xi_\bt, a\} \, \ud \nu = 0.
\end{equation}
Finally, noting that since $\xi_t \neq 0$ on the zero mass shell (and hence on the support of $\nu$), we see that the transport equation \eqref{eq:transport.def} holds for $g = m$ and for any $\widetilde{a}$ after using \eqref{eq:stupid.model.almost} $a = \f{\widetilde{a}}{\xi_t}$.

While in the above derivative, we used the Minkowski metric for simplicity, the \emph{same} argument works for a sequence of solutions to the linear wave equation with any \emph{fixed} metric after suitably collecting the extra terms (see \cite{Francfort, FrancfortMurat} for the stationary case, and \cite{GuerraTeixeira, HL.Burnett} for the general case).

Therefore, the main difficulty comes from the fact that $(\psi_i, \varpi_i)$ do not satisfy linear wave equations on fixed background. Instead, they satisfy \eqref{eq:psi.om.nonlinear} with both quasilinear terms from the metric and the semilinear terms  $\f 12 e^{-4\psi_i} g^{-1}(\ud \varpi_i, \ud \varpi_i)$, $4 g_i^{-1} (\ud \varpi_i, \ud \psi_i)$. As we will explain in the next three subsubsections, even though these extra nonlinear terms arise in the derivation of the transport equation \eqref{eq:transport.def}, the precise structure of these terms allows for \emph{compactness by compensation} so that they in fact do not contribute.


\subsubsection{Trilinear compensated compactness for three waves}\label{sec:Q0.trilinear}
The following structure is at the heart of handling the semilinear terms in the derivation of the transport equation \eqref{eq:transport.def}. In particular, it relies on the nonlinear terms in \eqref{eq:psi.om.nonlinear} being null forms of the type $Q_0(\phi^{(1)}, \phi^{(2)}) = g_0^{\alp\bt} \rd_\alp \phi^{(1)}_i \rd_\bt \phi^{(2)}_i$. Consider three sequences of functions $\phi^{(j)}_i$, with $j=1, 2, 3$ such that 
\begin{equation}\label{eq:trilinear.wave.a.priori.bound}
 \| \phi^{(j)}_i\|_{L^3} \to 0, \quad \| \rd \phi^{(j)}_i\|_{L^3} ,\, \|\Box_{g_0} \phi_i^{(j)} \|_{L^3}\lesssim 1.
\end{equation}
Fix a pseudo-differential operator $A$ of order $0$. A priori, the bound on $\| \rd \phi^{(j)}_i\|_{L^3}$ only shows that the term $\langle A (g_0^{\alp\bt} \rd_\alp \phi^{(1)}_i \rd_\bt \phi^{(2)}_i) ,\partial_t \phi^{(3)}_i \rangle_{L^2(\RR^2, g_0)}$ is $O(1)$. However, since $ g_0^{\alp\bt} \rd_\alp \phi^{(1)}_i \rd_\bt \phi^{(2)}_i = \f 12 \Box_{g_0}(\phi^{(1)}_i \phi^{(2)}_i) - \f 12\phi^{(1)}_i \Box_{g_0} \phi^{(2)}_i -  \f 12\phi^{(2)}_i \Box_{g_0} \phi^{(1)}_i$,
integrating by parts and using \eqref{eq:trilinear.wave.a.priori.bound} imply that the term in fact tends to $0$:
\begin{align*}
	\langle A (g_0^{\alp\bt} \rd_\alp \phi^{(1)}_i \rd_\bt \phi^{(2)}_i)),\partial_t \phi^{(3)}_i \rangle_{L^2(\RR^2, g_0)} =&\:  \f 12 \langle A\Box_{g_0}(\phi^{(1)}_i\phi^{(2)}_i)  ,\partial_t \phi^{(3)}_i \rangle_{L^2(\RR^2, g_0)} +o(1) \\
	=&\:  - \f 12 \langle A\partial_t (\phi^{(1)}_i\phi^{(2)}_i)  ,\Box_{g_0} \phi^{(3)}_i \rangle_{L^2(\RR^2, g_0)} + o(1) = o(1).
\end{align*}
As pointed out in \cite{GuerraTeixeira}, this simple observation can also be viewed as a consequence of the div-curl lemma in \cite{Murat.compensation, Tartar.compensation}.

\subsubsection{Elliptic-wave trilinear compensated compactness}\label{sec:elliptic.wave}
Another type of terms arising in \eqref{eq:psi.om.nonlinear} are the quasilinear terms $(\Box_{g_i} - \Box_{g_0})\psi_i$ and $(\Box_{g_i} - \Box_{g_0})\varpi_i$. In particular, in the derivation of the transport equation, we need to show the vanishing of some trilinear terms of the form
\begin{equation}\label{eq:difficult.elliptic.wave}
\langle [\rd A, g_i-g_0] \partial (\psi_i - \psi_0),\partial_t (\psi_i - \psi_0) \rangle_{L^2(\RR^{2+1},g_0)}.
\end{equation}

It is important to observe a few consequences of the elliptic gauge condition.
\begin{itemize}
\item Since the metric components $(\gamma_i, \bt_i, n_i)$ in \eqref{g.form.0} all satisfy elliptic equations (see \eqref{elliptic.1}--\eqref{elliptic.4}), their spatial derivatives $(\nabla \gamma_i, \nabla \bt_i, \nabla n_i)$ converge to $(\nabla \gamma_0, \nabla \bt_0, \nabla n_0)$ \emph{locally uniformly}, say
\begin{equation}\label{eq:spatial.improvement}
|\nabla (\gamma_i - \gamma_0, \bt_i - \bt_0, n_i - n_0)| \ls \lambda_i^{\f 12}.
\end{equation}
\item Using the equations, in fact $\rd_t \gamma_i$, $\rd_t \bt_i$ are also better and converge to their limits in \emph{norm}. In particular, $(\gamma_i, \bt_i)$ in fact converge to their limits in $C^1$.
\end{itemize}

Using these observations, if we have, say, $\gamma_i -\gamma$ in \eqref{eq:difficult.elliptic.wave}, then we can use Calder\'on commutator estimate (see, e.g., \cite[Corollary, p.309]{Stein}), \eqref{eq:burnett.assumption} and the $C^1$ convergence of $\gamma_i$ to bound
\begin{equation}
|\langle [\rd A, \gamma_i-\gamma_0] \partial (\psi_i - \psi_0),\partial_t (\psi_i - \psi_0) \rangle_{L^2(\RR^{2+1},g_0)}| \ls \|\rd (\psi_i - \psi_0) \|_{L^2}^2 \| \gamma_i-\gamma_0\|_{C^1} \to 0.
\end{equation}

The difficulty thus involves $n_i - n_0$, which does not converge in $C^1$. We consider the following model term:
\begin{equation}\label{comN}
 \langle [A,n_i-n_0] \Delta \phi_i , \partial_t \phi_i \rangle.
\end{equation}
The exact form of the term here will be important. To further simplify that exposition, let us assume that (1) $A$ is simply a Fourier multiplier with symbol $a(x, \xi) = m(\xi)$ independent of $x$, and (2) $\Box_{m} \phi_i$ (instead of $\Box_{g_0}\phi_i$ with a variable coefficient wave operator) obeys better bound $|\Box_{m} \phi_i| \ls 1$. Since $\phi_i$ is real-valued, we can also assume that $m$ is even.

Denoting $\xi = (\xi_t, \underline{\xi})$, $\eta = (\eta_t, \underline{\eta})$, we use Plancherel's identity to rewrite
\begin{equation}\label{eq:intro.hard.yet.again}
	\begin{split}
		\mbox{\eqref{comN}}=	&\: \int_{\mathbb R^{2+1}} \partial_t \psi_i \de^{j\ell} \{ (n_i-n_0) A \rd^2_{j\ell} \psi_i - A [(n_i - n_0) \rd^2_{j\ell} \psi_i]\} \,\ud x \\
		= &\: \frac{i}{2}\int_{\m R^{2+1}\times \m R^{2+1} }(\xi_t|\underline{\eta}|^2+ \eta_t|\underline{\xi}|^2)\widehat{(n_i-n_0)}(\eta-\xi)\widehat{\psi_i}(-\eta)\widehat{\psi_i}(\xi)(m(\xi)-m(\eta))\, \ud\xi \,\ud\eta.
	\end{split}
\end{equation}

Roughly speaking $ (\xi_t|\underline{\eta}|^2+ \eta_t|\underline{\xi}|^2)$ corresponds to three derivatives, and hence contributes to $O(\lambda_i^{-3})$ in size. This is sufficient to show that \eqref{comN} is bounded (by \eqref{eq:burnett.assumption}).  To deduce that \eqref{comN} actually tends to $0$, observe that
\begin{itemize}
	\item our main enemy is when $n_i-n_0$ has high $t$-frequency, i.e., $|\eta_t -\xi_t|$ is large (by \eqref{eq:spatial.improvement}), and
	\item we can gain with factors of $\underline{\xi} - \underline{\eta}$ (corresponding to spatial derivatives of $n_i -n_0$) or $\xi_t^2 - |\underline{\xi}|^2$ or $\eta_t^2 - |\underline{\eta}|^2$ (corresponding to $\Box_{g_0}$ acting on $\psi_i$).
\end{itemize}
Now the Fourier multiplier in \eqref{eq:intro.hard.yet.again} can be expressed as
$$ |\underline{\eta}|^2(\xi_t+\eta_t)=\eta_t(\underline{\xi}+\underline{\eta})(\underline{\xi}-\underline{\eta}) + |\underline{\eta}|^2 \frac{\xi_t^2-|\underline{\xi}|^2}{\xi_t-\eta_t}
+|\underline{\eta}|^2 \frac{|\underline{\eta}|^2-\eta_t^2}{\xi_t-\eta_t}+|\underline{\eta}|^2 \frac{(\underline{\xi}+\underline{\eta})\cdot (\underline{\xi} - \underline{\eta})}{\xi_t-\eta_t}.$$
When $\xi_t - \eta_t$ is large, we use the gain in $\xi_t^2-|\underline{\xi}|^2$, $|\underline{\eta}|^2-\eta_t^2$ or $(\underline{\xi} - \underline{\eta})$ to conclude that this term behaves better than expected.

Finally, to remove the assumption that $A$ and the wave operator are both constant-coefficient (in $x$), we ``freeze coefficients'' by localizing in sufficiently small balls.

We mention that in a subsequent work by Guerra--Teixeira da Costa \cite{GuerraTeixeira}, they introduced a slightly different argument to deal with these quasilinear term. Noting again that the time-frequency dominated part poses the most serious difficulty, they introduce frequency localization, take an inverse $\rd_t$ derivative, and then integrate by parts to reveal the cancellation. In particular, their argument avoids the computation in frequency space above.

\subsubsection{The wave map structure}\label{sec:wave.map.structure}
Finally, notice that in \eqref{eq:psi.om.nonlinear}, $\varpi_i$ appears in the $\Box_{g_i} \psi_i$ equation and vice versa. As a result, in the derivation of the transport equation, we need to handle the limit of a term of the form $\la \rd_t(\varpi_i -\varphi_0), g_0^{-1}(\ud \varpi_0, \ud (\psi_i - \psi_0) \ra_{L^2(\RR^{2+1},g_0)}$. Such a term would correspond to a ``cross'' microlocal defect measure between $\varpi_i -\varpi_0$ and $\psi_i - \psi_0$. However, because of the wave map structure, the total contribution of this type of terms would cancel in the derivation of the transport equation for $\nu$!

It turns out that the same phenomenon will play a role in the construction in Section~\ref{vlasov}; see Section~\ref{sec:parametrix.vlasov}. 


\section{Construction of small-amplitude high-frequency spacetimes in $\mathbb U(1)$ symmetry: multiple null dusts}\label{sec:U1.backward}

In this section, we discuss our work \cite{HL.HF} in which we construct high-frequency vacuum spacetimes in an elliptic gauge under $\m U(1)$ symmetry (see Section~\ref{sec:U1.gauge}) which converge in the limit to solutions to the Einstein--null dust system (see Remark~\ref{rmk:null.dust.as.vlasov}) with a finite but arbitrary number of families of dust.

Our construction can be viewed as a multiphase geometric optics construction. In particular, in our constructed high-frequency vacuum solutions, one sees that high-frequency waves propagating in different directions only interact with each other very weakly. Notice that in general multiphase geometric optics can be very complicated, with various possible phenomena of resonance and high-order harmonics creation; see for instance \cite{Met.book}. In our case, however, we have very good control of resonance and high-order harmonics. In fact, the nonlinear structure is sufficiently favorable that after suitable modifications of the methods, we can consider the limit as the number of phases goes to infinity; see~Section~\ref{vlasov}.

\subsection{Construction of high-frequency space-times: the case of null dusts}\label{dust}
The main result in \cite{HL.HF} approximates any small, regular and local in time solution to Einstein-null dust system in polarized $\m U(1)$ symmetry 
\begin{equation}\label{back}
	\left\{\begin{array}{l}
		\mathrm{Ric}_{\mu \nu}(g)= 2\partial_\mu \psi \partial_\nu \psi + \sum_{{\bA}} ({F}_{{\bA}})^2\partial_\mu u_{{\bA}} \partial_\nu u_{{\bA}},\\
		\Box_{g}\psi= 0,\\
		2(g^{-1})^{\alpha \beta}\partial_{\alpha} u_{{\bA}} \partial_{\beta} {F}_{{\bA}} + (\Box_{g} u_{{\bA}}) {F}_{{\bA}} = 0,\quad \forall \bA, \\
		(g^{-1})^{\alpha \beta}\partial_\alpha u_{{\bA}} \partial_\beta u_{{\bA}}=0,\quad \forall \bA,
	\end{array}
	\right.
\end{equation}
by a one parameter family of solutions to Einstein vacuum equations in polarized $\m U(1)$ symmetry
\begin{equation}\label{eq:U1vac:pol}
	\left\{
	\begin{array}{l}
		\mathrm{Ric}_{\alp\bt}(g)= 2\partial_\alp \psi \partial_\bt \psi, \\
		\Box_g \psi  = 0.	
	\end{array}
	\right.
\end{equation}

The sum in \eqref{back} on $\bA$ can be over an arbitrarily large but fixed finite set, which represents a finite number of families of null dusts propagating in different directions.



Before we state our main theorem, we note that in the small-data regime in $\mathbb U(1)$ symmetry, the constraint equations can be solved by specifying what we call \emph{admissible free initial data} \cite{Huneau.constraint, HL.elliptic} consisting of $(\dot{\psi}\doteq  \f{e^{2\gamma}}{n}(\rd_t \psi-\bt^i \rd_i \psi),\psi,\breve F_\bA\doteq F_{\bA} e^{\f\gamma 2}, u_\bA)\big|_{\{0\}\times \mathbb R^2}$ under an admissibility condition (see \cite{HL.elliptic} for details). Here, we implicitly assume that the initial hypersurface is maximal (i.e., with zero mean curvature). The following is the main theorem in \cite{HL.HF}, which is given in terms of admissible free initial data:
\begin{thm}[H.--L., Theorems~1.1, 4.2 in \cite{HL.HF}]\label{main.thm.2}
	Suppose $(\dot{\psi},\nabla \psi,\breve F_\bA, u_\bA)\big|_{\{0\}\times \mathbb R^2}$ is an admissible free initial data set satisfying the following:
	\begin{itemize}
		\item The level sets of $u_\bA$ are sufficiently close to planes and $u_\bA\big|_{\{0\}\times \mathbb R^2}$ is angularly separated, i.e., $\exists \eta' \in (0,1)$ such that 
		\begin{equation}\label{eq:ang.sep.def}
		\f{\delta^{ij}(\rd_i u_{\bA_1})(\rd_j u_{\bA_2})}{|\nabla u_{\bA_1}||\nabla u_{\bA_2}|}(t,x)<1-\eta', \quad \forall (t,x)\in I\times\mathbb R^2,\quad\forall \bA_1\neq \bA_2.
		\end{equation}
		\item  $\dot{\psi},\nabla \psi,\breve F_\bA$ are compactly supported and sufficiently small.
		\item A genericity condition holds for the initial data.
	\end{itemize}
	Then, 
	\begin{enumerate}
	\item a unique solution $(g_0,\psi_0, (F_0)_{\bA}, (u_0)_{\bA})$ to \eqref{back} arising from the given admissible free initial data set exists on a time interval $[0,1]$ in the elliptic gauge of Section~\ref{secell}, and 
	\item there exists a one-parameter family of solutions $(g_{\lambda},\psi_{\lambda})$ to \eqref{eq:U1vac:pol} for $\lambda \in (0,\lambda_0)$ (for some $\lambda_0\in \m R$ sufficiently small), which are all defined on the time interval $[0,1]$ in the elliptic gauge of Section~\ref{secell}, such that
	\begin{equation}\label{eq:main.thm.2.bd.1}
	(g_{\lambda},\psi_{\lambda}) \rightarrow (g_0,\psi_0)\quad \hbox{in $L^\infty_{\mathrm{loc}}$ and weakly in $H^1$}
	\end{equation}
	and 
	\begin{equation}\label{eq:main.thm.2.bd.2}
	\rd g_{\lambda}, \rd \psi_{\lambda}\in L^\infty_{\mathrm{loc}} \quad \hbox{uniformly in $\lambda$}.
	\end{equation}
	\end{enumerate}
\end{thm}

Let us make a few remarks on this theorem.
\begin{itemize}
	\item The existence and uniqueness of the limiting solution $(g_0,\psi_0, (F_0)_{\bA}, (u_0)_{\bA})$ is simply a local well-posedness result. This was proven in \cite{HL.elliptic}.
	\item The angular separation condition is used to control the interaction of null dusts propagating in different directions (see Section~\ref{sec:gain.high.frequency}).
	\item The genericity condition is a technical condition to ensure that we can take initial data for $(g_\lambda,\psi_\lambda)$ satisfying the orthogonality conditions necessary to solve the constraint equations with zero mean curvature.
\end{itemize}


From now on, we denote components of $g_\lambda$ (resp.~$g_0$) in an elliptic gauge schematically by $\mfg_\lambda$ (resp.~$\mfg_0$). The strategy of the proof, as explained in the following subsubsections, is to construct solutions to \eqref{eq:U1vac:pol} in elliptic gauge of the form
\begin{equation}
	\label{ansatz}\psi_{\lambda}= \psi_0+\sum_\bA \lambda F_\bA \cos(\tfrac{u_\bA}{\lambda})+\wht \psi_{\lambda},\quad \mfg_{\lambda}= \mfg_0 + \wht \mfg_{\lambda},
\end{equation}
where we used the shorthand $(F_\bA, u_{\bA})$ to denote $((F_0)_\bA, (u_0)_{\bA})$ as in the limiting solution, and $(\wht \psi_{\lambda}, \wht \mfg_{\lambda})$ are terms which are smaller $\lambda$, which in particular morally\footnote{The estimates for $\mfg$ as stated are not what we proved, since time and spatial derivatives behave differently in an elliptic gauge (see point (2) in Section~\ref{sec:dust.extra.difficulties}) and the lower order derivatives do not decay fast enough near spatial infinity. See \cite{HL.HF} for details.} satisfy\footnote{Moreover, we suppressed the necessary $\lambda$-independent smallness for clarity of exposition. Again, see \cite{HL.HF} for details.}
\begin{equation}\label{eq:norms.for.null.dust}
\sum_{k\leq 3} \lambda^{\max\{k-1,0\}} \| \rd^k \psi_\lambda \|_{L^\i_t L^2_x} \ls \lambda,\quad \sum_{k\leq 5} \lambda^{\max\{k-2,0\}} \| \rd^k \mfg \|_{L^\i_t L^2_x} \ls \lambda^2.
\end{equation}
The local well-posedness result in \cite{HL.elliptic} a priori gives existence of solutions to \eqref{eq:U1vac:pol} of the form \eqref{ansatz} on time intervals which are shrinking as $\lambda \to 0$, but we use a bootstrap argument on the remainder $(\wht \psi_\lambda, \wht g_{\lambda})$ to prove that the solutions remain regular and take the form \eqref{ansatz} up to time $1$.


If we manage to construct solutions to \eqref{eq:U1vac:pol} of the form \eqref{ansatz}, with the error terms satisfying the bounds \eqref{eq:norms.for.null.dust}, then the bounds \eqref{eq:main.thm.2.bd.2} in Theorem \ref{main.thm.2} immediately follow and moreover $(g_\lambda,\psi_\lambda)$ converges to the limiting solution $(g_0,\psi_0)$ in the sense of \eqref{eq:main.thm.2.bd.1}. 

Heuristically, one may think of the construction \eqref{ansatz} as a superposition of Burnett's example in the introduction. Moreover, it can be checked that if \eqref{ansatz} holds, then
\begin{equation}\label{eq:dpsi^2.limit}
\begin{split}
	&\:\partial_\alpha \psi_{\lambda} \partial_\beta \psi_\lambda \\
	=&\:\partial_\alpha \psi_{0} \partial_\beta \psi_0 + \sum_{\bA,\bB} F_{\bA} F_{\bB}  \sin (\tfrac{u_\bA}{\lambda}) \sin (\tfrac{u_\bA}{\lambda} ) \partial_\alpha u_\bA \partial_\beta u_\bB + \partial_{(\alpha|}\psi_0 \sum_{\bA}F_{\bA} \sin(\tfrac{u_\bA}{\lambda} )   \partial_{|\beta)} u_\bA +O_{L^2}(\lambda)\\
	=&\:\partial_\alpha \psi_{0} \partial_\beta \psi_0 + \frac{1}{2}\sum_{\bA} F_{\bA} ^2\left(1+ \cos(\tfrac{2u_\bA}{\lambda}) \right)  \partial_\alpha u_\bA \partial_\beta u_\bA + \sum_{\bA \neq \bB}F_{\bA} F_{\bB}  \cos(\tfrac{u_\bA\pm u_\bB}{\lambda} )\partial_\alpha u_\bA \partial_\beta u_\bB \\
	&\:\qquad \qquad+ \partial_{(\alpha|}\psi_0 \sum_{\bA}F_{\bA} \sin(\tfrac{u_\bA}{\lambda} )   \partial_{|\beta)} u_\bA +O_{L^2}(\lambda)\\
	\rightharpoonup &\:\partial_\alpha \psi_{0} \partial_\beta \psi_0 +\frac{1}{2}\sum_{\bA} F_{\bA} ^2\partial_\alpha u_\bA \partial_\beta u_\bA,
\end{split}
\end{equation}
which is why our constructed solutions to \eqref{eq:U1vac:pol} converge to a solution to \eqref{back}.

\subsubsection{Gaining from high-frequency} \label{sec:gain.high.frequency}

In order to justify \eqref{ansatz} and \eqref{eq:norms.for.null.dust}, it is important that we use the precise form of the data, and not only the bounds that they satisfy, since the regularity of the initial data are below the threshold known for general well-posedness. 

Before we proceed, we briefly review the standard mechanism behind geometric optics construction, which makes use of the high-frequency parameter as a smallness parameter; we refer the reader to \cite{Rauch.book} for a more thorough treatment of linear estimates. Since the system \eqref{eq:U1vac:pol} is a mixed hyperbolic-elliptic system in an elliptic gauge, we consider the model cases for both elliptic and wave equations, for which the behavior is quite different. For the discussion below, fix $\chi \in C^\infty_c$ and consider $\lambda \ll 1$. 
\begin{itemize}
\item When inverting an elliptic operator, say, solving $\Delta \phi = \chi \sin(\tfrac {x^i}\lambda)$, we have $\phi = O(\lambda^2)$. In fact, one can obtain a more precise expansion $\phi = \sum_{j=2}^{J} \lambda^j \phi_j + O(\lambda^{J+1})$, where  $\phi_j$ can be precisely computed and is bounded in, say, a norm $\sum_{k=0}^K \lambda^k \|\rd^k \phi_j \|_{L^p} \ls 1$. 
\item When inverting the wave operator, there is a difference depending on the direction of oscillation of the inhomogeneous terms. This can already be seen on Minkowski spacetime. Consider the Minkowskian problems:
$$\Box \phi = \chi \sin(\tfrac {t+x^i} \lambda),\quad \Box \phi = \chi \sin(\tfrac {t} \lambda), \quad \Box \phi = \chi \sin(\tfrac {x^i} \lambda).$$
It is well-known that in the first case, $\sum_{k=0}^K \lambda^k \| \rd^k \phi \|_{L^p} \ls \lambda$, and the leading order contribution satisfies a transport equation. In contrast, in the second and third cases, the behavior is effectively like in the elliptic case, and $\phi$ obeys the better estimates $\sum_{k=0}^K \lambda^k \| \rd^k \phi \|_{L^p}\ls \lambda^2$.
\end{itemize}
This already suggests that the error terms should be estimated in norms as in \eqref{eq:norms.for.null.dust}, where the bounds worsen by $\lambda^{-1}$ for each derivative.

Moreover, in order to exploit the above phenomena, we will need to arrange the phase appropriately:
\begin{itemize}
	\item First, we need that $u_\bA$, $u_{\bA} \pm u_{\bB}$ etc.~to oscillate with frequency $\sim \lambda^{-1}$ along the spatial directions. To achieve this we exploit a rescaling symmetry for the background: observe that if $(g,\psi,F_{\bA},u_{\bA})$ solves \eqref{back}, then for any set of positive constants $\{a_{\bA}\}_{\bA\in \mathcal A}\in \mathbb R_{>0}^{|\mathcal A|}$, $(g,\psi,F_{\bA}',u_{\bA}')$ is also a solution to \eqref{back} if we define
	$$F_{\bA}'=a_{\bA}^{-1} F_{\bA},\quad u_{\bA}'=a_{\bA} u_{\bA}.$$
	This rescaling allows us to ensure $|\nabla (a_\bA u_\bA \pm a_{\bB }u_{\gra B})|^2 \geq 1$ for phases corresponding to quadratic interaction, and similar lower bounds for higher order interactions.
	\item We will also need that for $\bA\neq \bB$, $u_{\bA} \pm u_{\bB}$ is either oscillating in a timelike or a spacelike direction, in order to capture the improved estimate discussed above. This cannot be arranged by rescaling alone, but is guaranteed by the angular separation assumption \eqref{eq:ang.sep.def}.

\end{itemize}

\subsubsection{First estimates for the metric coefficients}\label{sec:null.dust.metric.1st.est}
Recall that our goal will be to prove the estimates \eqref{eq:norms.for.null.dust}. We first discuss the estimates for $\mfg_\lambda$: we assume the estimates for $\psi_\lambda$ (but not that of $\mfg_\lambda$) in \eqref{ansatz}--\eqref{eq:norms.for.null.dust}, and try to derive suitable bounds for $\mfg_\lambda$. 

According to Lemma~\ref{lem:elliptic}, in the elliptic gauge, the metric components schematically satisfy equations of the form
$$\Delta \mfg_\lambda = \mathfrak L(g_\lambda)\Big[ \rd_\alp \psi_\lambda \rd_\bt \psi_\lambda \Big]+(\rd \mfg_\lambda)^2,\quad \Delta \mfg_0 = \mathfrak L(g_0)\Big[ \rd_\alp \psi_0 \rd_\bt \psi_0 + \f12 \sum_{\bA} F_{\bA} \rd_\alp u_{\bA} \rd_\bt u_{\bA} \Big]+(\rd \mfg_0)^2,$$
where $\mathfrak L(g)[\cdot]$ is some linear function of the unknown, with coefficients depending on $g$.

From the computations of $\partial_\alpha \psi_{\lambda} \partial_\beta \psi_\lambda $ in \eqref{eq:dpsi^2.limit} above, we see that the low-frequency part of $\mathfrak L(g_\lambda)\Big[ \rd_\alp \psi_\lambda \rd_\bt \psi_\lambda \Big]$ cancel exactly $\mathfrak L(g_0)\Big[ \rd_\alp \psi_0 \rd_\bt \psi_0 + \f12 \sum_{\bA} F_{\bA} \rd_\alp u_{\bA} \rd_\bt u_{\bA} \Big]$. Thus, assuming that $\psi_\lambda - \psi_0$ satisfy \eqref{ansatz}--\eqref{eq:norms.for.null.dust}, and introducing suitable bootstrap assumptions for $\mfg$, the equation for $\mfg_\lambda - \mfg_0$ then schematically takes the form
\begin{equation}\label{eq:null.dust.metric.diff}
\begin{split}
	\Delta (\mfg_\lambda-\mfg_0) = &\: \frac{1}{2}\sum_{\bA} F_{\bA} ^2 \cos (\tfrac{2u_\bA}{\lambda}) \partial u_\bA \partial u_\bA + \sum_{\bA \neq \bB}F_{\bA} F_{\bB}  \cos(\tfrac{u_\bA\pm u_\bB}{\lambda})\partial u_\bA \partial u_\bB \\
	&\: + \partial\psi_0 \sum_{\bA}F_{\bA} \sin ( \tfrac{u_\bA}{\lambda})   \partial u_\bA +O(\lambda).
\end{split}
\end{equation}
In this equation, the $O(1)$ terms have explicit expressions, and can be cancelled by introducing an approximate solution of the form
\begin{equation}\label{g1.def}
	\begin{split}
		\mfg_1=&-\f 18\sum_\bA \frac{\lambda^2 F_\bA^2}{|\nabla u_\bA|^2}  (\partial u_\bA)( \partial u_\bA) \cos (\tfrac{2u_\bA}{\lambda}) -\sum_\bA\frac{\lambda^2 F_\bA}{|\nabla u_\bA|^2} (\partial \psi_0 )(\partial u_\bA)  \sin (\tfrac{u_\bA}{\lambda}) \\
		&-\sum_{\bB \neq \bA}\frac{(\mp 1)\cdot \lambda^2 F_\bA F_{\gra B}}{|\nabla (u_\bA \pm u_{\gra B})|^2} (\partial_\mu u_\bA)( \partial_\nu u_{\gra B}) \cos ( \tfrac{ u_\bA \pm u_{\gra B}}{\lambda}).
	\end{split}
\end{equation}
In order to do so, we need lower bounds for $|\nabla u_\bA|^2$ and $|\nabla (u_\bA \pm u_{\gra B})|^2$, which can be arranged as explained in Section~\ref{sec:gain.high.frequency}. We can now write
\begin{equation}\label{eqg}
\Delta (\mfg - \mfg_0 - \mfg_1) =O(\lambda).
\end{equation}
In particular, the explicit form for $\mfg_1$ and standard elliptic estimates give the rough estimate
\begin{equation}\label{eq:null.dust.metric.1st.est}
\mfg - \mfg_0  = O_{H^2}(\lambda).
\end{equation}

\subsubsection{Need for a more precise parametrix}

It turns out that the ansatz \eqref{ansatz} is not precise enough to run our argument. To see this, we derive an equation for the error $\wht \psi_{\lambda}$ by the following computations:
\begin{equation}\label{eq:null.dust.error.terms}
	\begin{split}
	0= \Box_{g_\lambda }\psi_\lambda= &(\Box_{g_{\lambda}}-\Box_{g_0})\psi_0
	-\frac{1}{\lambda}\sum_{\bA }	 g_\lambda^{\alpha \beta}\partial_\alpha u_\bA \partial_\beta u_\bA  F_\bA \cos(\tfrac{u_\bA}{\lambda})\\
	&-\sum_{\bA} (2g^{\alpha\beta}_{\lambda}\partial_\alpha u_\bA \partial_\beta F_\bA + \Box_{g_{\lambda}} u_\bA F_\bA)\sin (\tfrac{u_\bA}{\lambda})
	+\lambda \sum_{\bA }	 \Box_{g_\lambda} F_\bA \cos (\tfrac{u_\bA}{\lambda})
	+\Box_{g_\lambda}\wht \psi_{\lambda},
	\end{split}
\end{equation}
where we used $\Box_{g_0} \psi_0 = 0 = \Box_{g_\lambda }\psi_\lambda$. Equation \eqref{eq:null.dust.error.terms} can be viewed as an equation for $\Box_{g_\lambda}\wht \psi_{\lambda}$, which can be used to estimate $\wht \psi_{\lambda}$.

We first note that the $(\Box_{g_{\lambda}}-\Box_{g_0})\psi_0$ term, the $\sum_{\bA} (2g^{\alpha\beta}_{\lambda}\partial_\alpha u_\bA \partial_\beta F_\bA + \Box_{g_{\lambda}} u_\bA F_\bA)\sin (\tfrac{u_\bA}{\lambda})$ term and the $\lambda \sum_{\bA } \Box_{g_\lambda} F_\bA \cos (\tfrac{u_\bA}{\lambda})$ term all satisfy $\sum_{k=0}^K\lambda^k \|\rd^k \cdot \|_{L^2}\ls \lambda$. For the first term, this is due to  \eqref{eq:null.dust.metric.1st.est}, for the second term, this is due to the transport equation for $F_\bA$ in \eqref{back} (together with \eqref{eq:null.dust.metric.1st.est}); for the final term, this is due to the extra $\lambda$ present. These estimates are already sufficient for proving the wave estimate (i.e., the first bound) in \eqref{eq:norms.for.null.dust}.

Thus the main error term in \eqref{eq:null.dust.error.terms} is the second term on the right-hand side. Recalling that $u_\bA$ satisfies the eikonal equation with respect to the metric $g_0$, i.e., $g_0^{\alp\bt} \rd_\alp u_\bA \rd_\bt u_{\bB} = 0$ (see \eqref{back}), we can write
\begin{equation}\label{eq:hard.term.parametrix}
\frac{1}{\lambda} g_\lambda^{\alpha \beta}\partial_\alpha u_\bA \partial_\beta u_\bA=\frac{1}{\lambda}(g_\lambda^{\alpha \beta}-g_0^{\alpha \beta})\partial_\alpha u_\bA \partial_\beta u_\bA.
\end{equation}
If we only have \eqref{eq:null.dust.metric.1st.est}, then the term \eqref{eq:hard.term.parametrix} is $O_{L^2}(1)$, which only gives the $\| \rd \wht \psi_{\lambda}\|_{L^2} \ls 1$, and is in turn too weak to justify \eqref{eq:null.dust.metric.diff}.

\subsubsection{Improved parametrix}\label{sec:dust.improved.para}
Due to considerations outlined above, we need a more precise parametrix than \eqref{ansatz}. Instead, we further expand $\psi_\lambda$ as follows:
\begin{equation}\label{psi.para}
	\begin{split}
		\psi_\lambda= &\psi_0+\sum_\bA \lambda F_\bA \cos (\tfrac{u_\bA}{\lambda}) +\sum_\bA \lambda^2 \wht F_\bA \sin(\tfrac{u_\bA}{\lambda})\\
		&+\sum_\bA \lambda^2 \wht F_\bA^{(2)}\cos(\tfrac{2u_\bA}{\lambda})+\sum_\bA  \lambda^2\wht F_\bA^{(3)} \sin(\tfrac{3u_\bA}{\lambda }) + \mathcal E_\lambda,
	\end{split}
\end{equation} 
where $\wht F_\bA$, $\wht F_\bA^{(2)}$ and $\wht F_\bA^{(3)}$ are $O(1)$ terms which are defined to satisfy suitable transport equations, and $\mathcal E_\lambda$ is an error term which satisfies the even better estimate $\sum_{k\leq 3} \lambda^{\max\{k-1,0\}} \| \rd^k \mathcal E_\lambda \|_{L^\infty_t L^2_x} \ls \lambda^2$.

The key point here is that \eqref{psi.para} is sufficiently precise to keep track of the $O(\lambda)$ contribution in  $\Delta (\mfg_\lambda - \mfg_0 - \mfg_1)$ in \eqref{eqg}, and show that they are in fact also of high-frequency. As a result, we can obtain a more accurate the approximate solution for $\mfg_\lambda$ by introducing a $\mfg_2$ term (in a similar manner as \eqref{g1.def}) with size $\sum_{k=0}^K \lambda^k \| \rd^k \mfg_2\|_{L^p} \ls \lambda^{3}$ so that $\Delta (\mfg_\lambda - \mfg_0 - \mfg_1 - \mfg_2) = O_{L^2}(\lambda^2)$. As a result, we obtain $\mfg_\lambda - \mfg_0 = O(\lambda^2)$, which is sufficient to handle the difficult term \eqref{eq:hard.term.parametrix}.

Let us note that in order to justify the parametrix in \eqref{psi.para}, we in turn need to be more precise with the $O(\lambda)$ error terms in the equation \eqref{eq:null.dust.error.terms} (i.e., terms such as $(\Box_{g_{\lambda}}-\Box_{g_0})\psi_0$). In particular, we need to define $\wht F_\bA$, $\wht F_\bA^{(2)}$ and $\wht F_\bA^{(3)}$ suitably to cancel with these error terms. 

Finally, notice that the extra terms in \eqref{psi.para} are high frequency terms that only involve the phases $u_\bA$, $2 u_\bA$ and $3 u_\bA$. These terms correspond to either linear terms in the high-frequency source or terms representing (quadratic or cubic) parallel interactions. The reason that we do not need to keep track of high-frequency terms e.g., with a phase $u_{\bA} \pm u_{\bB}$ with $\bA \neq \bB$ is because these phases oscillate either in the timelike or in the spacelike direction, so that the output is better according to considerations in Section~\ref{sec:gain.high.frequency}. In particular, these terms are $O(\lambda^3)$ (instead of $O(\lambda^2)$) and can be considered as part of $\mathcal E_\lambda$.

\subsubsection{Additional technical issues}\label{sec:dust.extra.difficulties}
We end with two more technical difficulties that arise in the proof.
\begin{enumerate}
	\item (Coupling between various terms) It should be emphasized that in the bootstrap argument, not only the main error terms $\mfg_\lambda - \mfg_0 - \mfg_1 - \mfg_2$ and $\mathcal E_\lambda$ couple. In fact, the way the paramatrix is defined dictates that the estimates for these also have to be coupled with that of $\mfg_2$ and $\wht F_{\bA}$. In particular, we need to contend with a potential loss of derivatives. (In fact, for related reasons, it would seem difficult to obtain a parametrix which is more precise than \eqref{psi.para}.)
	\item (Time derivative of the metric coefficients) Another difficulty comes from estimating $\rd_t \mfg$: the $\rd_t$ derivatives of the metric components are worse than the spatial derivatives because the metric components solve spatial elliptic equations on fixed time slices. Here, it is important that the $\rd_t$ derivatives of some specific components behave better (cf.~the structure used in Section~\ref{sec:elliptic.wave}). Moreover, there is an important cancellation coming from the two uncontrollable terms $\rd_t (\mfg_\lambda - \mfg_0 - \mfg_1 - \mfg_2)$ and $\rd_t^2 \wht F_{\bA}$.
\end{enumerate}

\section{Construction of small-amplitude high-frequency spacetimes in $\mathbb U(1)$ symmetry: from null dusts to massless Vlasov}\label{vlasov}
While we constructed a large class of high-frequency vacuum spacetimes which limit to solutions to the Einstein--null dust system (see Section~\ref{sec:U1.backward}), Conjecture~\ref{conj:backward} is not restricted to the Einstein--null dust system. In this section, we discuss ongoing work in which we construct examples where the limiting spacetimes are more general solutions to the Einstein--massless Vlasov system. In particular, in the notations of Definition~\ref{def:the.final.def}, $\nu$ is now allowed to be absolutely continuous with respect to the Lebesgue measure.

Our goal now is to construct, in the $(2+1)$-dimensional space $[0,1]\times \mathbb R^2$, a sequence of solutions $(g_i,\psi_i,\varpi_i)$ to (recall \eqref{eq:U1vac})
\begin{equation}\label{eq:U1vac.again}\left\{
	\begin{array}{l}
		\Box_g \psi + \frac{1}{2}e^{-4\psi}g^{-1}(\ud\varpi,\ud\varpi)  = 0,\\
		\Box_g \varpi -4g^{-1}(\ud\varpi,\ud\psi)=0,\\
		\mathrm{Ric}_{\mu \nu}(g)= 2\partial_\mu \psi \partial_\nu \psi + \frac{1}{2}e^{-4\psi}\partial_\mu \varpi \partial_\nu \varpi
	\end{array}
	\right.
\end{equation}
which converges to $(g_0,\psi_0, \varpi_0)$, where $(g_0,\psi_0, \varpi_0, \{ f(\omega),u(\omega)\}_{\om \in \mathbb S^1})$ is a solution to
\begin{equation}\label{vlasovsys}
	\left\{\begin{array}{l}
		\Box_g \psi + \frac{1}{2}e^{-4\psi}g^{-1}(\ud\varpi,\ud\varpi)  = 0,\\
		\Box_g \varpi -4g^{-1}(\ud\varpi,\ud\psi)=0,\\
		\mathrm{Ric}_{\mu \nu}(g)= 2\partial_\mu \psi \partial_\nu \psi +\frac{1}{2}e^{-4\psi}\partial_\mu \varpi \partial_\nu \varpi+ \int_{\m S^1}f^2(t,x,\omega)\partial_\mu u(t,x,\omega) \partial_\nu u(t,x,\omega)\,\ud m(\omega),\\
		2(g^{-1})^{\alpha \beta}\partial_{\alpha} u \partial_{\beta} f + (\Box_{{g}} u) f = 0 \quad\forall \omega \in \mathbb S^1,\\
		(g^{-1})^{\alpha \beta}\partial_\alpha u \partial_\beta u=0,\quad u \big|_{\{t=0\}} = x\cdot \omega,\, \rd_t u \big|_{\{t=0\}} >0,\quad \forall \om \in \mathbb S^1,
	\end{array}
	\right.
\end{equation}
where $\ud m(\omega)$ is a fixed probability measure on $\m S^1$.
In what follows, we denote $U \doteq (\psi,\varpi)$ and $\langle \partial_\alpha U, \partial_\beta U\rangle \doteq \partial_\alpha \psi \partial_\beta \psi + \frac{1}{4}e^{-4\psi}\partial_\alpha \varpi \partial_\beta \varpi$. 

Here, in \eqref{vlasovsys}, we consider a specific class of solutions to the Einstein--Vlasov system, where the cotangent bundle can be parametrized by $\{\ud u(\om)\}_{\omega \in \mathbb S^1}$ so that the connection to the Einstein--null dust system (see Section~\ref{sec:U1.backward}) is more transparent. In particular, this parametrization allows for Vlasov field which is absolutely continuous with respect to the Lebesgue measure. Notice that the systems \eqref{eq:U1vac.again} and \eqref{vlasovsys} in $(2+1)$-dimensions arise, respectively, as reductions of the Einstein vacuum equations and the Einstein--massless Vlasov system in $(3+1)$ dimensions under $\mathbb U(1)$ symmetry.

We now give a rough statement of our main theorem. As in Theorem~\ref{main.thm.2}, we state the theorem in terms of admissible free initial data, where suitable modifications need to be adapted in this setting. As before, we assume that the initial hypersurface is maximal. 
\begin{thm}[H.--L., to appear]\label{thm:main.vlasov}
	Suppose $(\dot{\psi},\nabla \psi,\dot{\varpi}, \nabla\varpi,\{\breve f(\om)\}_{\om\in \mathbb S^1}, \{u(\om)\}_{\om\in \mathbb S^1}) \big|_{\{0\}\times \mathbb R^2}$ is an admissible free initial data set satisfying the following:
	\begin{itemize}
		\item The level sets of $u(\om)$ are chosen to be planes with angle $\om$ to the $y$-axis in the sense that $\nabla u(\omega) = (\cos\om, \sin\om)$.
		\item	$(\dot{\psi},\nabla \psi,\dot{\varpi}, \nabla\varpi,\{\breve f(\om)\}_{\om\in \mathbb S^1})$ are compactly supported and sufficiently small.	
		\item A genericity condition holds.
	\end{itemize}
	Then 
	\begin{enumerate}
	\item a unique solution $(g_0,\psi_0, \varpi_0, \{f_0(\om)\}_{\om\in\mathbb S^1}, \{u_0(\om)\}_{\om\in\mathbb S^1})$ to \eqref{vlasovsys} arising from the given admissible free initial data set exists on a time interval $[0,1]$ in the elliptic gauge of Section~\ref{secell}, and 
	\item there exists a sequence of solutions $\{(g_{i},\psi_{i},\varpi_{i})\}_{i=1}^\infty$ to \eqref{eq:U1vac.again}, which are all defined on the time interval $[0,1]$ in the elliptic gauge of Section~\ref{secell}, such that
	\begin{equation}\label{eq:main.thm.vlasov.bd.1}
	(g_{i},\psi_{i},\varpi_{i}) \rightarrow (g_0,\psi_0,\varpi_0)\quad \hbox{in $L^\infty_{\mathrm{loc}}$ and weakly in $H^1$}
	\end{equation}
	and 
	\begin{equation}\label{eq:main.thm.vlasov.bd.2}
	\rd g_{i}, \rd \psi_{\lambda}\in L^4_{\mathrm{loc}} \quad \hbox{uniformly in $i\in \mathbb N$}.
	\end{equation}
	\end{enumerate}
\end{thm}

Notice that Theorem~\ref{thm:main.vlasov} simultaneously generalizes Theorem~\ref{main.thm.2} in two ways: in addition to allowing for more general Vlasov field, we remove the polarization assumption and allow for general $\mathbb U(1)$-symmetric solutions so that we deal with a $(2+1)$-dimensional Einstein--wave map system (see Remark~\ref{rmk:twist}).

Some comments about the basic strategy of the proof of Theorem~\ref{thm:main.vlasov} are in order.
\begin{itemize}
	
	\item The strategy of the proof consists of two steps. First, we approximate the solution to the Einstein--massless Vlasov system by a sequence of solutions to Einstein--null dust system (see Section~\ref{sec:vlasov.by.dust}), where the number of families of dust $\to \infty$. We then use a construction similar to that in Section~\ref{sec:U1.backward} to approximate solutions to the Einstein--null dust system by solutions to \eqref{eq:U1vac} (see Section~\ref{sec:parametrix.vlasov}).
	
	\item However, we emphasize that Theorem~\ref{thm:main.vlasov} does not follow from Theorem~\ref{main.thm.2}. Indeed, in Theorem~\ref{main.thm.2}, as the number of families of dust $N$ increases, we required more stringent smallness assumption. As a result, we cannot directly pass to the $N\to\infty$ limit in Theorem~\ref{main.thm.2}.

	\item To carry out the proof, we track the dependence on $N$ (the number of families of dust) in the argument \cite{HL.HF}. We then need a modification of the argument so that $\ep$, the size of $(g_0,U_0, f, u)$, is independent of $\frac{1}{N}$. However, we will assume that $\lambda$ is small compared to some function of $N$ to obtain extra smallness.
\end{itemize}

\subsection{Approximation of a Vlasov field by null dusts}\label{sec:vlasov.by.dust}
Using that the set of convex combinations of Dirac measures is weak-* dense in the set of all probability measures, we construct a particular  weak-* approximating sequence as follows.
Let $m$ be a given probability measure on $\mathbb S^1 \doteq \mathbb R/(2\pi \mathbb Z)$. For all $N \in \mathbb N$, and $\bA = 0,1,\cdots,N-1$, we can find $N$ separated points $\om_\bA^{(N)} = \f{\bA}{2\pi N} \in \m S^1$, and $N$ coefficients $\sigma_\bA^{(N)} = m\Big([\f{\bA}{2\pi N}, \f{\bA+1}{2\pi N})\Big)$ (with $\sigma_\bA^{(N)} \geq 0$ and $\sum_{\bA=0}^{N-1} \sigma_\bA^{(N)} =1$) such that
\begin{equation}\label{eq:m.limit}
	\sum_{\bA=0}^{N-1} \sigma_\bA^{(N)} \delta_{\omega_{\bA}^N} \overset{\ast}{\rightharpoonup} m,
\end{equation}
in the weak-* topology as $N \to \infty$.
To approach \eqref{vlasovsys} by $N$ dusts, we consider the initial data for \eqref{vlasovsys} $(\psi,\partial_t \psi, F^\psi(\omega),F^\varpi(\omega))$. We then solve the coupled system
\begin{equation}\label{backdeux}
	\left\{\begin{array}{l}
		\mathrm{Ric}_{\mu \nu}(g)= 2 \langle \partial_\mu U, \partial_\nu U \rangle + \sum_{{\bA}} ((F^\psi_{{\bA}})^2+\frac{e^{-4\psi}}{4}(F^\varpi_\bA)^2)\partial_\mu u_{{\bA}} \partial_\nu u_{{\bA}},\\
		\Box_g \psi + \frac{1}{2}e^{-4\psi}g^{-1}(d\varpi,d\varpi)  = 0,\\
		\Box_g \varpi -4g^{-1}(d\varpi,d\psi)=0,\\
		2(g^{-1})^{\alpha \beta}\partial_{\alpha} u_\bA \partial_{\beta} F^{\psi}_\bA + (\Box_{{g}} u_\bA) F^\psi_\bA + e^{-4\psi}(g^{-1})^{\alpha \beta}\partial_\alpha \varpi \partial_\beta u_\bA F^\varpi_\bA = 0,\\
		2(g^{-1})^{\alpha \beta}\partial_{\alpha} u_\bA \partial_{\beta} F^{\varpi}_\bA + (\Box_{{g}} u_\bA) F^\varpi_\bA -4(g^{-1})^{\alpha \beta}\partial_\alpha \varpi \partial_\beta u_\bA F^\psi_\bA -4(g^{-1})^{\alpha \beta}\partial_\alpha \psi \partial_\beta u_\bA F^\varpi_\bA=0,\\
		(g^{-1})^{\alpha \beta}\partial_\alpha u_{{\bA}} \partial_\beta u_{{\bA}}=0.
	\end{array}
	\right.
\end{equation}
where we have dropped the subscript $N$ and denoted $u_{{\bA}} =u(\omega_{\bA})$ and $F^{\psi,\varpi}_{{\bA}}=\sigma_{\bA}^{\f 12} f^{\psi,\varpi}(\omega_{\bA})$. 

Notice that in \eqref{backdeux}, instead of having a single transport equation for the null dust as in \eqref{back}, we have introduced a decomposition of the density of the null dust as $F^2_\bA= (F^\psi_{{\bA}})^2+\frac{e^{-4\psi}}{4}(F^\varpi_\bA)^2$, and the two transport equations for $F^\psi_{{\bA}}$ and $F^\varpi_\bA$ in \eqref{backdeux} implies the following transport equation:
$$2(g^{-1})^{\alpha \beta}\partial_{\alpha} u_\bA \partial_{\beta} F_\bA + (\Box_{{g}} u_\bA) F_\bA=0.$$
The importance of this decomposition will become clear in the next section.

We consider solutions to \eqref{backdeux} with high regularity. Using a compactness argument, we can show the convergence of a subsequence towards a solution of \eqref{vlasovsys} as $N \to \infty$. It thus remains to approximate solutions to \eqref{backdeux} by solutions to \eqref{eq:U1vac.again}

\subsection{Parametrix for $\psi$ and $\varpi$}\label{sec:parametrix.vlasov}
In order to approach a solution of \eqref{backdeux} by solutions to \eqref{eq:U1vac.again}, we show the existence of solutions to \eqref{eq:U1vac.again} of the form
\begin{equation}\label{ansatz2}
	\psi_\lambda = \psi_0 + \sum_\bA\lambda  F^\psi_{\bA}\cos\left(\frac{ u_{\bA}}{\lambda}\right)+ \wht \psi,\quad \varpi_\lambda = \varpi_0 + \sum_\bA\lambda F^\varpi_{\bA}\cos\left(\frac{u_{\bA}}{\lambda}\right)+ \wht \varpi,\quad \mfg_\lambda=\mfg_0+\wht \mfg_\lambda,
\end{equation}
where $\wht \psi,\, \wht \varpi,\, \wht \mfg_\lambda $ are $O(\lambda^2)$.

The construction is similar to the one explained in Section~\ref{sec:U1.backward}. In particular, we need a more precise parametrix than \eqref{ansatz2} in order to capture nonlinear interactions in parallel directions (see Section~\ref{sec:dust.improved.para}). Here, we point out a few new difficulties corresponding to the $N\to \infty$ limit.
\begin{itemize}
	\item There is already a difficulty to obtain a local existence result, even on a time interval depending on $\lambda$, for initial data consistent with \eqref{ansatz2}. Indeed, due to the globality (in space) of the elliptic estimates for the metric components, local existence requires a smallness assumption, and a straightforward extension\footnote{We recall that strictly speaking \cite{HL.elliptic} only applies to the polarized case, though the proof carries over to the non-polarized case.} of the results in \cite{HL.elliptic} would require smallness for $\|\partial \psi\|_{L^\infty}$ and $\|\partial \varpi\|_{L^\infty}$. On the other hand, for the main high-frequency term in $\sum_\bA  F^\psi_{\bA} \partial u_\bA \cos\left(\frac{ u_{\bA}}{\lambda}\right)$, since we only have smallness for $\Big(\sum_\bA  |F^\psi_{\bA}|^2 + |F^\varpi_{\bA}|^2\Big)^{\f 12}$ but not for $\sum_\bA \Big( |F^\psi_{\bA}|^2 + |F^\varpi_{\bA}| \Big)$, it is unclear that $\|\partial \psi\|_{L^\infty}$ and $\|\partial \varpi\|_{L^\infty}$ are small. Nonetheless, almost orthogonality of the high-frequency phases allows one to conclude that $\|\partial \psi\|_{L^4}$ is small when $\lambda$ is sufficiently large with respect to $N$ (compare \eqref{eq:main.thm.vlasov.bd.2} with \eqref{eq:main.thm.2.bd.2}). Because of this, we need to use the more recent improved local existence result by Touati \cite{Touati.local} (see Remark~\ref{rmk.touati.local}).	
	\item A lot of terms in our construction can only be estimated with large constants that grow with $N$. As a result, the corresponding $O(\lambda^2)$ terms in Section~\ref{sec:U1.backward} are now only required to obey an $N$-dependent bound $\lambda^2 e^{A(N)t}$, where $A(N)$ grows polynomially in $N$. 
	\item A priori, one potential danger would be that error terms of size $\lambda^2 e^{2A(N)t}$ arise from nonlinear interactions, i.e., the $N$ dependence becomes worse for nonlinear interactions. It turns out that this does not occur since the extra exponential growth in $N$ can be compensated by additional factors of $\lambda$. However, in order to achieve this, we need to work with phases $u_{\bA}$ in the parametrix that solve the eikonal equation for the perturbed metric $g_\lambda$ instead of that for the background metric $g_0$. 
	\item In order to control the $u_{\bA}$ satisfying the true eikonal equation, it will be useful to control the corresponding null second fundamental form $\chi_{\bA} \doteq \Box_{g_\lambda} u_{\bA}$ using the Raychaudhuri equation. To carry out the estimates, a parametrix decomposition has to be introduced, both for $u_{\bA}$ and $\chi_{\bA}$.
	\end{itemize}
Finally, there is also an additional difficulty coming from consider the general (instead of polarized) $\mathbb U(1)$ symmetry:
	\begin{itemize}
	\item The presence of the semilinear nonlinear terms on the right-hand side in the wave equations for $\psi$ and $\varpi$ requires us to modify our ansatz. The transport equations satisfied by $F_\bA^\psi$ and $F_\bA^\varpi$ in \eqref{backdeux} are exactly designed so that the zeroth order terms (in powers of $\lambda$) to vanish. Similar adaptations are also necessary at higher order.
	Importantly, thanks to the null form, no harmonics are generated at the first order. 
\end{itemize}

\section{High-frequency spacetimes in generalized wave coordinates}\label{sec:wave.coordinates}

In this section, we consider high-frequency spacetimes in generalized wave coordinates. 

\begin{definition}
Let $n \geq 2$. We say that $(\mathcal M^{n+1}, g)$ satisfies the \textbf{generalized wave coordinates condition} if
\begin{equation}\label{eq:gen.wave.coord}
g^{\mu\nu} \Gamma_{\mu\nu}^\alp = H^\alp,\quad \forall \alp.
\end{equation}
In the special case that $H^\alp \equiv 0$, we call \eqref{eq:gen.wave.coord} the \textbf{wave coordinates condition}.
\end{definition}

The Einstein vacuum equations in generalized wave coordinates can be written
\begin{equation}\label{gw}\Box_g g_{\mu \nu}= Z_{\mu \nu}(\partial g, \partial g) + g_{(\mu|\rho}\partial_{|\nu)}H^\rho,
	\end{equation}
and therefore takes the form of a system of quasilinear wave equations.

It is well-known, particularly due to the work of Lindblad--Rodnianski \cite{LinRod.WN, LinRod} on the global stability of Minkowski spacetime in wave coordinates, that the nonlinearity has a special \emph{weak null} structure, which is weaker than the classical null condition, but is still much more special than generic quadratic derivative nonlinearities. This particular structure --- both the presence of a weak null structure and the failure of the classical null condition --- is important for Burnett's conjecture.

\subsection{Burnett's conjecture in generalized wave coordinates}\label{sec:wave.Burnett}

In order to simplify the exposition, we focus on the case where the wave coordinate condition holds, i.e., when $H^\alp \equiv 0$ $\forall \alp$ in \eqref{eq:gen.wave.coord}. The theorem easily generalizes to the case where $H_i^\alp \to H_0^\alp$ is a suitably strong (but still weaker than $C^1$) topology; see \cite{HL.wave} for details.

\begin{theorem}[H.--L. \cite{HL.wave}]\label{thm:Burnett.wave}
Burnett's conjecture (with high frequency condition \eqref{eq:intro.HF.def} for $K = 2$) holds if we assume in addition that the wave coordinate condition holds for $g_0$ and for $g_i$ for all $i\geq 1$.
\end{theorem}

As for Burnett's conjecture in elliptic gauge under $\m U(1)$ symmetry (see Section~\ref{sec:U1.forward}), the proof of the theorem gives a precise description of the massless Vlasov field in the limit, which is related to a suitably defined microlocal defect measure. Also as in Section~\ref{sec:U1.forward}, the key to Theorem~\ref{thm:Burnett.wave} is the precise structure of the linear and nonlinear terms. From the work of Lindblad--Rodnianski \cite{LinRod.WN}, it is known that the nonlinear terms in the Ricci curvature tensor in wave coordinates do not satisfy the classical null condition. In fact, the terms which fail the null condition can be identified:
$$\mathrm{Ric}_{\mu\nu}(g) = - \f 12 \widetilde{\Box}_g g_{\mu\nu} + \f 12 P_{\mu\nu}(g)(\rd g, \rd g) + \mbox{terms satisfying null condition},$$
where
\begin{equation}\label{eq:tBox.def}
\widetilde{\Box}_g q_{\mu\nu} \doteq g^{\alp\bt} \rd^2_{\alp\bt} q_{\mu\nu},
\end{equation}
and 
\begin{equation}\label{eq:P.def}
P_{\mu\nu}(g)(\rd p, \rd q) \doteq \frac 14 g^{\alp\alp'}\rd_\mu p_{\alp \alp'} g^{\bt \bt'}\rd_\nu q_{\bt \bt'} - \frac 12 g^{\alp\alp'}\rd_\mu p_{\alp \bt} g^{\bt \bt'}\rd_\nu q_{\alp' \bt'}.
\end{equation}

In order to prove Theorem~\ref{thm:Burnett.wave}, we need the linear and nonlinear structures of $\mathrm{Ric}_{\mu\nu}(g_0) - \mathrm{Ric}_{\mu\nu}(g_i)$:
\begin{lemma}\label{lem:wave.h.equation}
Define $h_i = g_i - g_0$. Assume that
\begin{enumerate}
\item the wave coordinates condition holds $g_i$ and $g_0$,
\item $g_i, g_0, g_i^{-1}, g_0^{-1}$ and their first derivatives are uniformly bounded, and
\item $g_i^{-1} - g_0^{-1} = o_{i\to \infty}(1)$.
\end{enumerate}

Then 
\begin{equation}\label{eq:Ric.form}
\begin{split}
 \mathrm{Ric}_{\mu\nu}(g_0)   
= &\: \mathrm{Ric}_{\mu\nu}(g_i) + \f 12 \widetilde{\Box}_{g_0} (h_i)_{\mu\nu} - \f 12 g_0^{\alp\alp'} g_0^{\bt\bt'} (h_i)_{\alp'\bt'} \rd^2_{\alp\bt} (h_i)_{\mu\nu} - \f 12 L_{\mu\nu}(g_0)(\rd h_i) \\
&\: - \f 12 P_{\mu\nu}(g_0)(\rd h_i, \rd h_i) 
+ \mbox{quadratic terms in $\rd h_i$ satisfying null condition} + o_{i\to \infty}(1),
\end{split}
\end{equation}
where $\tBox_{g_0}$ and $P_{\mu\nu}$ are as in \eqref{eq:tBox.def} and \eqref{eq:P.def}, respectively, and the linear term $L_{\mu\nu}$ is given by 
\begin{equation}\label{eq:L.def}
L_{\mu\nu}(g_0)(\rd h) \doteq 4 g_0^{\sigma\rho}\Gamma_{\rho}{}^{\alp}{}_{(\mu|}(g_0)\rd_{\sigma} h_{|\nu)\alp} + D_{(\mu|}^{\alp\sigma} (g_0)\rd_{|\nu)} h_{\alp\sigma}
\end{equation}
and 
\begin{equation}
D_{\mu}^{\alp\sigma}(g_0) \doteq g_0^{\alp\bt} g_0^{\sigma\rho} (2 \rd_\rho (g_0)_{\bt \mu} - \rd_{\mu} (g_0)_{\bt\rho}).
\end{equation}
\end{lemma}


Using the structure of the equation in Lemma~\ref{lem:wave.h.equation}, it is not difficult to see that the limiting Ricci curvature tensor must take the following form:
\begin{proposition}\label{prop:wave.Ricci}
$\mathrm{Ric}_{\mu\nu}(g_0)$ is given by 
$$\int_{\mathcal M} \psi \mathrm{Ric}_{\mu\nu}(g_0) \, \mathrm{dVol}_g = \int_{\mathcal S^* \mathcal M} \xi_\mu\xi_\nu \psi\, \ud \mu,\quad \forall \psi \in C^\infty_c(\mathcal M),$$
where 
\begin{equation}\label{def:wave.coord.mu}
\mu = g_0^{\alp\rho} g_0^{\bt\sigma} (\f 14 \mu_{\rho\bt\alp\sigma} - \f 18 \mu_{\rho\alp\bt\sigma}),
\end{equation}
and $\mu_{\alp\bt\rho\sigma}$ are the microlocal defect measure defined (similarly as Proposition~\ref{prop:nu}) so that (after passing to a subsequence)
$$\la \rd_\gamma (h_i)_{\alp\bt}, A \rd_\de (h_i)_{\rho\sigma} \ra_{L^2} \to \int_{\mathcal S^* \RR^{3+1}} a \, \ud \mu_{\alp\bt\rho\sigma}$$
for any zeroth order pseudodifferential operator $A$ with principal symbol $a$.
\end{proposition}
\begin{proof}
This proposition amounts to computing the weak limit of the right-hand side of \eqref{eq:Ric.form}. First, notice that by assumption, $\mathrm{Ric}_{\mu\nu}(g_i) = 0$. Next, note that a similar argument as in \eqref{eq:null.CC.Qab} and \eqref{eq:null.CC.Q0} shows that the terms satisfying the null condition do not contribute to the limit. It thus suffices to consider terms written out in \eqref{eq:Ric.form}.

For the four remaining terms, notice that $\f 12 \widetilde{\Box}_{g_0} (h_i)_{\mu\nu}$ and $\f 12 L_{\mu\nu}(g_0)(\rd h_i)$ are both linear in $h_i$ and/or its derivatives, and thus have weak limit $ = 0$.

For the quasilinear term, we observe that 
\begin{equation}\label{eq:wave.Burnett.quasilinear}
\begin{split}
&\:  g_0^{\alp\alp'} g_0^{\bt\bt'} (h_i)_{\alp'\bt'} \rd^2_{\alp\bt} (h_i)_{\mu\nu} \\
= &\: \rd_\alp \Big(g_0^{\alp\alp'} g_0^{\bt\bt'} (h_i)_{\alp'\bt'}  \rd_{\bt} (h_i)_{\mu\nu}\Big) + g_0^{\alp\alp'} g_0^{\bt\bt'} \rd_{\alp} (h_i)_{\alp'\bt'} \rd_\bt (h_i)_{\mu\nu} + o(1),
\end{split}
\end{equation}
where we have used \eqref{eq:intro.HF.def}. Now the first term in \eqref{eq:wave.Burnett.quasilinear} is a total derivative of an $o(1)$ term, which tends to $0$ weakly. For the second term, we note that the wave coordinate condition allows us to rewrite it as 
$$\f 12 g_0^{\alp\alp'} g_0^{\bt\bt'} \rd_{\bt'} (h_i)_{\alp\alp'} \rd_\bt (h_i)_{\mu\nu} + o(1),$$
so that the first term contains an exact null form which therefore has a vanishing weak limit!

It thus follows that the only possibly non-zero limit comes from the term $\f 12 P_{\mu\nu}(g_0)(\rd h_i, \rd h_i)$ in \eqref{eq:Ric.form}. The definition of $\mu$ exactly captures the contribution of this term. \qedhere
\end{proof}

Just as in Theorem~\ref{thburnett}, the most difficult part of Burnett's conjecture is therefore the following transport equation for $\mu$:
\begin{theorem}\label{thm:wave.Burnett.key}
The following holds for $\mu$ defined in \eqref{def:wave.coord.mu}:
$$\int_{\mathcal S^* \mathcal M} \{g_0^{\mu\nu} \xi_\mu\xi_\nu, \widetilde{a} \} \, \ud \mu = 0, \forall \widetilde{a}$$
\end{theorem}

Combining Proposition~\ref{prop:wave.Ricci} and Theorem~\ref{thm:wave.Burnett.key} then yields Theorem~\ref{thm:Burnett.wave} (cf.~the $\mathbb U(1)$-symmetric case in Definition~\ref{def:the.final.def}).

It must be emphasized that the proof of Theorem~\ref{thm:wave.Burnett.key} strongly relies on the fact we are considering $\mu$ in Theorem~\ref{thm:wave.Burnett.key}. The analogous statement is not expected to hold for each individual $\mu_{\alp\bt\rho\sigma}$ as we have used cancellations coming from the combination in \eqref{def:wave.coord.mu}; cf.~Section~\ref{sec:wave.map.structure}.

The starting point of the proof of Theorem~\ref{thm:wave.Burnett.key} is an ``energy estimate'' type computation in the spirit of Section~\ref{sec:Mink.warmup}, which gives the following propagation equation for $\mu$:
\begin{equation}\label{eq:wave.Burnett.main}
\begin{split}
\int_{\calS^*\RR^{3+1}} \{g_0^{\mu\nu} \xi_\mu\xi_\nu, \widetilde{a}(x,\xi) \}  \, \ud \mu = &\: \f 14\int_{\calS^*\RR^{3+1}} g_0^{\mu\nu} \xi_\nu \widetilde{a}(x,\xi)  \rd_{x^\mu} (2g_0^{\alp\alp'} g_0^{\bt\bt'} - g_0^{\alp\bt} g_0^{\alp'\bt'}) \, \ud \mu_{\alp\bt\alp'\bt'} \\
&\: + \f 12 \lim_{i\to \infty}  \la (2 g_0^{\alp\alp'} g_0^{\bt\bt'} - g_0^{\alp\bt} g_0^{\alp'\bt'}) \rd_{t} (h_i)_{\alp\bt}, A\widetilde{\Box}_{g_0} (h_i)_{\alp'\bt'} \ra.
\end{split}
\end{equation}

The proof of Theorem~\ref{thm:wave.Burnett.key} now boils down to showing that the right-hand side of \eqref{eq:wave.Burnett.main} $ \equiv 0$. This is again a compensated compactness type argument: we plug in the equation for $\widetilde{\Box}_{g_0} (h_i)_{\alp'\bt'} $ from Lemma~\ref{lem:wave.h.equation} (with $\mathrm{Ric}(g_i) = 0$, $\forall i \geq 1$) and investigate each term.
\begin{enumerate}
\item Since $\mathrm{Ric}(g_0)$ is smooth and $i$-independent, its contribution $=0$ using that $\rd_t(h_i)_{\alp\bt} \rightharpoonup 0$ weakly.
\item The linear term $L_{\mu\nu}$ in \eqref{eq:Ric.form} can be computed exactly using the microlocal defect measures. It turns out that there is an  algebraic cancellation where the contribution from $L_{\mu\nu}$ cancels \emph{exactly} the first term on the right-hand side of \eqref{eq:wave.Burnett.main}.
\item For the terms involving null forms, we show that none of them contribute to the limit. For $Q_0$, this can be achieved by a simple integration by parts argument as in Section~\ref{sec:Q0.trilinear}. For $Q_{\alp\bt}$, we use the trilinear normal form estimates of Ionescu--Pasauder \cite{IP} (introduced in their proof of the stability of Minkowski spacetime for the Einstein--Klein--Gordon system). 
\item For the term $P$ which fails the classical null condition, there is a hidden null structure in the \emph{trilinear} term which can be revealed after integration by parts:
\begin{equation*}
\begin{split}
&\: \la g_0^{\alp\alp'} g_0^{\bt\bt'} \rd_t (h_i)_{\alp'\bt'}, A ( g_0^{\rho\rho'} g_0^{\sigma\sigma'} \rd_\alp (h_i)_{\rho\sigma} \rd_\bt (h_i)_{\rho'\sigma'} )\ra \\
= &\: \la g_0^{\alp\alp'} g_0^{\bt\bt'} (\rd_t A^* (h_i)_{\alp'\bt'}),  g_0^{\rho\rho'} g_0^{\sigma\sigma'} \rd_\alp (h_i)_{\rho\sigma} \rd_\bt (h_i)_{\rho'\sigma'} \ra + o(1) \\
= &\: \la g_0^{\alp\alp'} g_0^{\bt\bt'} (\rd_\alp A^* (h_i)_{\alp'\bt'}),  g_0^{\rho\rho'} g_0^{\sigma\sigma'} \rd_t (h_i)_{\rho\sigma} \rd_\bt (h_i)_{\rho'\sigma'} \ra \\
&\: - \la g_0^{\alp\alp'} g_0^{\bt\bt'} Q_{t\alp}(A^* (h_i)_{\alp'\bt'},(h_i)_{\rho\sigma}) ,  g_0^{\rho\rho'} g_0^{\sigma\sigma'} \rd_\bt (h_i)_{\rho'\sigma'} \ra + o(1) \\
= &\: \f 12\la g_0^{\alp\alp'} Q_0 (A^* (h_i)_{\alp\alp'}, (h_i)_{\rho'\sigma'}),  g_0^{\rho\rho'} g_0^{\sigma\sigma'} \rd_t (h_i)_{\rho\sigma} \ra \\
&\: - \la g_0^{\alp\alp'} g_0^{\bt\bt'} Q_{t\alp}(A^* (h_i)_{\alp'\bt'},(h_i)_{\rho\sigma}) ,  g_0^{\rho\rho'} g_0^{\sigma\sigma'} \rd_\bt (h_i)_{\rho'\sigma'} \ra + o(1),
\end{split}
\end{equation*}
where in the second step, we swapped the $\rd_t$ and $\rd_\alp$ derivative at the expense of a null form; in the third step, we used the wave coordinate condition. Finally, since all these terms consist of null forms, they vanish in the limit.
\item Finally, for the quasilinear term $g_0^{\mu\mu'} g_0^{\nu\nu'}(h_i)_{\mu\nu} \rd^2_{\mu'\nu'} (h_i)_{\alp'\bt'}$, we first observe that the main difficulty arises when $(h_i)_{\mu\nu}$ has high frequency and that the frequency lives near the light cone of $g_0$:
\begin{itemize}
\item If $(h_i)_{\mu\nu}$ has low frequency, then we integrate by parts and use Calder\'on commutator estimates to force a derivative to act on $(h_i)_{\mu\nu}$. Since $(h_i)_{\mu\nu}$ has low frequency, this leads to an improvement.
\item If $(h_i)_{\mu\nu}$ has high frequency, but the frequency lives away from the light cone, then we use that we have estimates for $\Box_{g_0} h_i$ (because of the Einstein equation) and that $\Box_{g_0}$ is elliptic for such a frequency regime.
\end{itemize}
We thus concentrate to the high-frequency regime near the light cone. We can write the frequency localized part as a total $\rd_t$ derivative, i.e., $(h_i)^{\mathrm{freq.~loc.}}_{\mu\nu} = \rd_t (\mathfrak k_i)_{\mu\nu}$, where $\sum_{j\leq 1}\lambda^{j}\| \rd^j (\mathfrak k_i)_{\mu\nu} \|_{L^2} \ls \lambda^{1+b}$ and $\|\Box_{g_0} \mathfrak k_i\|_{L^2} \ls \lambda^{b}$ for some $b \in (\f 12, 1)$. Integrating by parts, we obtain
\begin{equation}
\begin{split}
&\: \la g_0^{\alp\alp'} g_0^{\bt\bt'}  \rd_\gamma (h_i)_{\alp\bt}, A (g_0^{\mu\mu'} g_0^{\nu\nu'}\rd_t (\mathfrak k_i)_{\mu\nu} \rd^2_{\mu'\nu'} (h_i)_{\alp'\bt'} )\ra \\
= &\: \la g_0^{\alp\alp'} g_0^{\bt\bt'}  \rd_\gamma (h_i)_{\alp\bt}, A (g_0^{\mu\mu'} g_0^{\nu\nu'} \rd_{\mu'} (\mathfrak k_i)_{\mu\nu} \rd^2_{t\nu'} (h_i)_{\alp'\bt'} )\ra \\
&\:+  \la g_0^{\alp\alp'} g_0^{\bt\bt'}  \rd_\gamma (h_i)_{\alp\bt}, A (g_0^{\mu\mu'} g_0^{\nu\nu'} Q_{t\mu'}((\mathfrak k_i)_{\mu\nu}, \rd_{\nu'} (h_i)_{\alp'\bt'}) )\ra \\
= &\: \f 12 \la g_0^{\alp\alp'} g_0^{\bt\bt'} \rd_\gamma (h_i)_{\alp\bt}, A (g_0^{\mu\mu'} g_0^{\nu\nu'} \rd_{\nu} (\mathfrak k_i)_{\mu\mu'} \rd^2_{t\nu'} (h_i)_{\alp'\bt'} )\ra \\
&\:+  \la g_0^{\alp\alp'} g_0^{\bt\bt'} \rd_\gamma (h_i)_{\alp\bt}, A (g_0^{\mu\mu'} g_0^{\nu\nu'}  Q_{t\mu'}((\mathfrak k_i)_{\mu\nu}, \rd_{\nu'} (h_i)_{\alp'\bt'}) )\ra + o(1)\\
= &\: \f 12 \la g_0^{\mu\mu'} g_0^{\alp\alp'} g_0^{\bt\bt'}  \rd_\gamma A^*(h_i)_{\alp\bt},  Q_0( (\mathfrak k_i)_{\mu\mu'}, \rd_{t} (h_i)_{\alp'\bt'}) \ra \\
&\:+  \la g_0^{\mu\mu'} g_0^{\nu\nu'} g_0^{\alp\alp'} g_0^{\bt\bt'} \rd_\gamma A^*(h_i)_{\alp\bt},  Q_{t\mu'}((\mathfrak k_i)_{\mu\nu}, \rd_{\nu'} (h_i)_{\alp'\bt'}) \ra + o(1),
\end{split}
\end{equation}
where in the first step we exchanged $\rd_t$ and $\rd_{\mu'}$ at the expense of a null form, in the second step we used that the wave coordinate condition for $h$ implies a good bound for $H(g_0)(\rd \mathfrak k)$\, and in the third step we noted that the commutation of $A^*$ with $g_0^{\mu\mu'} g_0^{\alp\alp'} g_0^{\bt\bt'}  \rd_\gamma$ and $g_0^{\mu\mu'} g_0^{\nu\nu'} g_0^{\alp\alp'} g_0^{\bt\bt'} \rd_\gamma$ are in $\Psi^{-1}$. As before, we have thus obtained null forms in every term.
\end{enumerate}

See \cite{HL.wave} for details of the proof.

\subsection{The geometric optics approximation with one phase}\label{sec:touati}
In \cite{Touati2}, Touati considered the geometric optics solutions with one phase in generalized wave coordinates in $(3+1)$-dimensions. This could be thought of as an analogue of the results in Section~\ref{sec:U1.backward}, without any symmetry assumptions, but restricted only to one phase.

More precisely, Touati showed in \cite{Touati2} the existence of solutions to Einstein vacuum equations in generalized wave coordinates \eqref{gw} of the schematic form
$$g_\lambda = g_0 + \lambda g^{(1)} \left(\frac{u_0}{\lambda}\right) + \widetilde{g}_\lambda,$$
where $g_0$ is a metric in wave coordinates and is a small-data solution to Einstein-null dust equation:
\begin{equation}
	\left\{\begin{array}{l}
		\mathrm{Ric}_{\mu \nu}= F_0^2\partial_\mu u_0 \partial_\nu u_0\\
		g_0^{-1}(du_0,du_0)=0\\
		2 g_0^{\alpha \beta}\partial_\alpha u_0 \partial_\beta F_0 +\Box_{g_0} u_0 = 0.
\end{array}
\right.
\end{equation}
The tensor $g^{(1)} \left(\frac{u_0}{\lambda}\right) $ is given by $\cos\left(\frac{u_0}{\lambda}\right)F_{\mu \nu}^{(1)}$with 
\begin{align}
	\label{pol} g_0^{\mu \nu}(\partial_\mu u_0 F^{(1)}_{\sigma \nu}-\frac{1}{2}\partial_\sigma u_0 F^{(1)}_{\mu \nu})=&\: 0,\\
	\frac{1}{8}|F^{(1)}|^2_{g_0}-\frac{1}{16}(\mathrm{tr}_{g_0}F^{(1)})^2= &\: F_0^2, \\
	\label{transport}2\partial^\alpha u_0 D_\alpha F^{(1)}_{\mu \nu}+(\Box_{g_0} u_0)F_{\mu \nu}^{(1)}= &\:0.
	\end{align}
This ansatz correspond to the one which was formally computed by Choquet-Bruhat in \cite{CB.HF}. 
The proof of Touati is based on the following elements.
\begin{itemize}
	\item As in \cite{HL.HF}, one needs a more precise expansion than is given above, and there is a need to go up to second order in the ansatz.
	\item The proof uses generalized wave coordinates, where $H^\alp$ in \eqref{eq:gen.wave.coord} is $O(\lambda)$ small but not zero. The choice of  relies on the polarization condition \eqref{pol} and is used to absorb the unwanted harmonics. The fact that this is possible is in some sense a similar aspect as the elliptic equations satisfied by the metric coefficients in $\m U(1)$ symmetry. 
	\item There is an apparent loss of derivative in the construction of the high-frequency ansatz, due to the quasilinear nature of Einstein equations, which was dealt with using a clever frequency cutoff.
	\item It can already be seen in the first order ansatz that $F^{(1)}$ should satisfy both a polarization condition  \eqref{pol} and a transport equation \eqref{transport}. The fact that they are compatible is a computation, already present in \cite{CB.HF}. Similar conditions are present for the second order ansatz. The compatibility of these conditions, intractable by computation, is proved using the Bianchi identities.
\item The high-frequency ansatz, with terms satisfying transport equations, and the used of generalized wave coordinate condition require a special construction of the initial data. This is done in \cite{Touaticontraintes}.
\end{itemize}

\subsection{Superposition of high frequency waves}\label{sec:touati.new}
This geometric optic construction have been extended by Touati in the very recent work \cite{touati2024reverse} to construct exact solutions to Einstein vacuum equations which can be written as a superposition of high-frequency waves
$$g_\lambda = g_0 + \lambda\sum_\bA \cos\left(\frac{u_\bA}{\lambda}\right)F_\bA^{(1)}+ \widetilde{g}_\lambda,$$
 This is an analogue of the result in \cite{HL.HF}, but now without any symmetry assumptions. 
The fact that high frequency waves propagating in different null direction produce only a quadratic error in $\lambda$ is not a priori straightforward and is again a manifestation of the structure of the nonlinearity in Einstein equations.

\section{Future directions and open problems}\label{sec:future}

\subsection{Alternative characterizations of limit spacetimes}

Burnett's conjecture (Conjecture~\ref{conj:forward}) is phrased in terms of a massless Vlasov field that is yet to be determined. In practice, in all the known results that we have surveyed, there is a natural candidate for the massless Vlasov field, which is a microlocal defect measure associated with the failure of convergence. Nevertheless, it could be useful to find a characterization of the limit spacetime in terms of the metric itself.

\begin{problem}
Given a smooth metric $g$ with an Einstein tensor $G(g) \doteq \mathrm{Ric}(g) - \f12 Rg$ which is non-negative definite (i.e., $G(X,X) \geq 0$ for any vector $X$) and trace-free. Find suitable criteria to determine whether $g$ can be viewed as a solution to the Einstein--massless Vlasov system after introducing a massless Vlasov field.
\end{problem}

\subsection{Analogue of Burnett's conjecture in lower regularity}

In \eqref{eq:intro.HF.def}, uniform pointwise estimates are assumed for the derivative of the metric, and even though second derivatives are allowed to grow, they are assumed to grow in a particular fashion with good pointwise control. However, the question in Conjecture~\ref{conj:forward} already makes sense for $g_i \to g$ in $C^0$ and weakly in $H^1$. It is of interest to understand whether the Burnett conjecture continues to hold with the weaker notion of convergence.
\begin{problem}\label{prob:low.regularity}
Does the Burnett conjecture still hold if we only assume $g_i \to g$ in $C^0$ and weakly in $H^1$?
\end{problem}

The question in Problem~\ref{prob:low.regularity} has an affirmative answer in the setting of angularly regular spacetimes; see Section~\ref{sec:angularly.regular}. However, other results (Theorems~\ref{thburnett} and \ref{thm:Burnett.wave}) rely on stronger assumptions.
Weak convergence in $H^1$ can be thought of as natural for two reasons: (1) it is the weakest known regularity for a notion of weak solution to make sense \cite{GerochTraschen}, and (2) weak convergence in $H^{1+\epsilon}$ for $\epsilon>0$ immediately implies that the limit is also vacuum (and hence the Burnett conjecture becomes trivial); see Section~\ref{sec:intro.backward}. 

Concerning Problem~\ref{prob:low.regularity}, one can in fact already ask a simpler question in the setting of wave maps.
\begin{problem}
Let $d=2,3$ and $(\mathcal N, h)$ be a Riemannian manifold. Suppose $\{\Phi_i\}_{i=1}^\infty$ is a sequence of smooth wave maps $\Phi_i: \mathbb R^{d+1} \to \mathcal N$ such that $\Phi_i$ converges to a limiting smooth map $\Phi: \mathbb R^{d+1} \to \mathcal N$ in $C^0$ and weakly in $H^1$. Does a suitably-defined microlocal defect measure characterizing the convergence satisfies the massless Vlasov equation?
If the answer is negative, it would also be of interest to understand the exact regularity threshold for the failure.
\end{problem}

\subsection{Issue of gauge} In all the formulations above, a specific \emph{gauge} is fixed. It is of interest to understand whether any of the results can be formulated in a gauge-independent manner. A simpler question would be to understand whether the conclusion of Burnett's conjecture still holds after ``high-frequency gauge transformation.''
\begin{problem}
Is the Burnett conjecture gauge dependent? In particular, given a sequence of solution $(\mathcal M,g_i)$ to the Einstein vacuum equations, and assume that $g_i \to g_0$ according to \eqref{eq:intro.HF.def} for some $g_0$ satisfying the Einstein--massless Vlasov system (with a suitable Vlasov field). Introduce a sequence of new coordinates $\{y^\alp_i\}_{i=1}^\infty$ such that with respect to the original coordinates $\{x^\bt\}$ the following bounds are satisfied:
$$ \Big| \f{\rd y^\alp_i}{\rd x^\bt} - \de^\alp_\bt \Big| \ls \lambda_i,\quad \Big| \f{\rd^2 y^\alp_i}{\rd x^\bt\rd x^{\bt'}} \Big| \ls 1.$$
Is the limit in the new coordinates still a solution to the Einstein--massless Vlasov system?
\end{problem}

\subsection{Global solutions}

The constructions concerning Conjecture~\ref{conj:backward} in this survey are all \emph{local-in-time} results. (The only known global constructions require $\mathbb T^2$ symmetry, see Section~\ref{sec:related.examples}.) In view of the known results on the stability of Minkowski spacetime both in vacuum\footnote{Related to the constructions in $\mathbb U(1)$, we also note that the stability of Minkowski spacetime in vacuum is known under $\mathbb U(1)$ symmetry \cite{Huneau.stability}, despite the data being not asymptotically flat when viewed as data in $3+1$ dimensions.} \cite{CK, LinRod} and for the Einstein--massless Vlasov system \cite{BFJST, mT2017}, one may expect it to be possible to construct global examples, at least in a neighborhood of Minkowski spacetime.

\begin{problem}
Construct a family of \textbf{global} (in the sense of geodesically complete) vacuum spacetimes in a neighborhood of Minkowski spacetime with high-frequency oscillations so that the limit corresponds to a global spacetime satisfying the Einstein--massless Vlasov system.
\end{problem}

Perhaps the simplest global constructions could come from outgoing high-frequency pulses constructed in a similar manner as in the semi-linear problem considered in \cite{Touati3}. Near null infinity, one may consider these outgoing high-frequency pulses in a double null coordinate gauge; for this, the ideas in \cite{Ang20} on semi-global impulsive gravitational waves may be relevant.

\subsection{Large solutions}

The only construction that we have which allows for large data is in the angularly regular setting of Section~\ref{sec:angularly.regular}. In particular, the constructions discussed in Sections~\ref{sec:U1.backward}, \ref{vlasov} and \ref{sec:touati} all used that the solutions are close to Minkowski. Notice that in the $\m U(1)$ setting, because of our use of an elliptic gauge, smallness is required even for local existence of \emph{smooth} solution. It would be interesting to carry out these constructions without smallness assumptions:
\begin{problem}
Construct local high-frequency solutions as in Sections~\ref{sec:U1.backward}, \ref{vlasov} and \ref{sec:touati} but such that the limiting solution is \textbf{far away from Minkowski spacetime}.
\end{problem}

\subsection{Geometric optics for infinite number of families of null dusts without symmetry}

A natural problem that arises from Touati's work \cite{touati2024reverse} (see Section~\ref{sec:touati.new}) is to extend the results in \cite{touati2024reverse} to geometric optics solutions with an infinite number of phases. This can be viewed as an analogue of the results in Section~\ref{vlasov} so that the limiting spacetime has a Vlasov field which is absolutely continuous with respect to the Lebesgue measure, but in $(3+1)$-dimensions without symmetry assumptions.
\begin{problem}
Construct geometric optics solutions to the Einstein vacuum equations (in $(3+1)$ dimensions without any symmetry) with infinitely many phases so that the high-frequency limit corresponds to solutions to the Einstein--massless Vlasov system where the Vlasov field is absolutely continuous with respect to the Lebesgue measure.
\end{problem}


\subsection{Geometric optics beyond caustics}

All the constructions so far rely on geometric optics type constructions where that the null hypersurfaces remain well-controlled. It would be very interesting to go beyond this and to study geometric optics beyond caustics in the nonlinear setting:
\begin{problem}
Construct high-frequency geometric optics type solutions to the Einstein vacuum equations \textbf{beyond caustics}.
\end{problem}
See \cite{Duistermaat, Hormander71, Ludwig, Maslov} for this type of constructions for linear equations. The corresponding nonlinear theory is much less developed, and appears quite far to be applicable to the Einstein equations, but we refer the reader to \cite{Carles, HunKel87, JMR96} for some related results.

\subsection{Burnett's conjecture from the initial data point of view} While Burnett's conjecture is primarily about the \emph{dynamics} of the Einstein equations, one can also consider the question on a fixed spacelike hypersurface, and ask about the behavior of high-frequency limits of solutions to the constraint equations. This could be slightly simpler since the constraint equations can be thought of as being elliptic.

More precisely, consider a sequence $\{ (\hat{g}_i, \hat{k}_i) \}_{i=1}^\infty$ (where $\hat{g}_i$ are Riemannian metrics and $\hat{k}_i$ are symmetric covariant $2$-tensors) satisfying the vacuum constraint equations
\begin{equation}\label{eq:constraints}
R(\hat{g}_i) - |\hat{k}_i|^2_{\hat{g}_i} + \mathrm{tr}_{\hat{g}_i} \hat{k}_i = 0\quad \hat{\nabla}_j \hat{k}_\ell{}^j - \hat{\nabla}_\ell \mathrm{tr}_{\hat{g_i}} \hat{k}_i = 0,
\end{equation}
where $R$ denotes the scalar curvature, $\hat{\nabla}$ denotes the Levi-Civita connection of $\hat{g}_i$ and indices are raised with respect to $\hat{g}_i$.
We would like to understand the following problem:
\begin{problem}
Classify all limits of suitable ``high-frequency solutions'' $(\hat{g}_i, \hat{k}_i)$ to \eqref{eq:constraints}.
\end{problem}

Already one can ask the question when $\hat{k}_i \equiv 0$ for every $i \in \mathbb N$. In this case, the constraint equations \eqref{eq:constraints} reduce to simply $R(\hat{g}_i) = 0$. It is known by the works of Gromov \cite{Gro.Burnett} and Bamler \cite{Bam.Burnett} that even if $\hat{g}_i$ only has a $C^0$ limit $\hat{g}_0$, the limit must satisfy $R(\hat{g}_0) \geq 0$. This is consistent with Conjecture~\ref{conj:forward} as the limit must satisfy the weak energy condition. In the spirit of Conjecture~\ref{conj:backward}, one may also ask whether all non-negative scalar curvature metrics arise as $C^0$ limits of scalar-flat metrics. For this, we refer the reader to a related result of Lohkamp \cite[Theorem~B]{Lohkamp} which shows that the set of metrics with non-positive scalar curvature is $C^0$-dense in the set of all metrics. This may motivative the following conjecture:
\begin{conjecture}
Let $\mathcal M$ be a manifold of dimension $\geq 3$. Suppose $\hat{g}_0$ is a smooth metric on $\mathcal M$ with $R(\hat{g}_0) \geq 0$. Then there exists a sequence of smooth metrics $\{\hat{g}_i\}_{i=1}^\infty$ on $\mathcal M$ with $R(\hat{g}_i) = 0$ that converge to $\hat{g}_0$ in $C^0_{\mathrm{loc}}$.
\end{conjecture}

\subsection{From Einstein--massless Vlasov back to Einstein vacuum}

One interesting prospect of the questions surrounding the Burnett conjecture is the possibility of understanding the vacuum equations using
\begin{itemize}
\item information about the Einstein--massless Vlasov system, and
\item the bridge between the Einstein vacuum system and the Einstein--massless Vlasov system via the Burnett conjecture.
\end{itemize}
This is particularly useful when the situation is considerably simpler with a Vlasov field, possibly because one can impose spherical symmetry or because some explicit computations can be done. We give a few examples below: Section~\ref{sec:trapped.surfaces} is an example that is already discussed in \cite{LR.HF}, while Sections~\ref{sec:ads} and \ref{sec:geons} are some possible future directions.

\subsubsection{Formation of trapped surfaces}\label{sec:trapped.surfaces} The celebrated incompleteness theorem of Penrose \cite{Penrose.2} shows that the presence of trapped surfaces, together with suitable energy conditions of the matter field and non-compactness of initial data, must imply that the maximal Cauchy development of the initial data is geodesically incomplete. However, the theorem of Penrose does not address the question of the \textbf{dynamical formation} of trapped surfaces. This was particularly difficult in the vacuum case, where spherical symmetry cannot be imposed (because of Birkhoff's theorem), and one is forced to deal with a long-time, large-data regime for the Einstein vacuum equations outside symmetry. This problem was finally resolved in the monumental work of Christodoulou \cite{Chr}, in which he constructed spacetimes in which a ``short pulse'' of gravitational waves focus to create a trapped surface. Here, the short pulse data is concentrated on a short length scale of size $\de$. In the proof, Christodoulou made use of the short length scale to propagate a hierarchy of $\de$-dependent estimates  so that the estimates can be closed despite the solution being in a large-data regime.

As was pointed out in \cite{LR.HF}, in the $\de\to 0$ limit, the construction of Christodoulou coincides with the dynamical trapped surface formation scenario with a null dust shell. While the limiting procedure requires hard analysis to justify (and needs the estimates in \cite{Chr}), the trapped surface formation mechanism in the null dust shell case is much easier to understand, and provides a simpler conceptual model for Christodoulou's construction.

%

\subsubsection{Instability of anti-de Sitter spacetime}\label{sec:ads}

The anti-de Sitter spacetime is a vacuum solution to the Einstein equations with a negative cosmological constant. It is conjectured \cite{eguchiha} to be unstable under reflective boundary conditions, and this has been studied heuristically and numerically \cite{bizon2011weakly}. (Note, however, that it is expected to be globally nonlinear asymptotically stable under \emph{dissipative} boundary condition \cite{dissipative}.)

The instability of the anti-de Sitter spacetime for the Einstein vacuum equations is still at present out of reach. However, in a recent breakthrough \cite{Mos23}, Moschidis proved that the anti-de Sitter spacetime is unstable to trapped surface formation for the Einstein--massless Vlasov system. (See also \cite{Mos20} for results on the null dust model with an inner mirror). 

Given the relation between the Einstein vacuum equations and the Einstein--massless Vlasov system as proposed by Burnett's conjectures, it is natural to ask to what extent \cite{Mos23} sheds light on the instability problem for the anti-de Sitter spacetime in vacuum. (Notice that there is an obvious issue to directly apply the results \cite{Mos23} here, namely that the proof in \cite{Mos23} relies on a well-posedness result in a very weak topology, which is only expected to hold in spherical symmetry.)

%

\subsubsection{Gravitational geons and static solutions to the Einstein--massless Vlasov system}\label{sec:geons}

Another class of interesting solutions to the Einstein--massless Vlasov system are \emph{static} solutions which neither disperse nor collapse into a black hole. These solutions have been constructed in spherical symmetry in \cite{AFT}. On the other hand, this kind of static solutions are not expected to exist in vacuum. In view of Burnett's conjectures, it is of interest to construct high-frequency vacuum solutions that weakly approximate these static solutions, at least for a long time. This is related to the gravitational geons of Brill--Hartle \cite{BrillHartle}, where high-frequency gravitational waves propagate in a confined region on a background geometry, which is created by the effective stress-energy-momentum of the waves themselves. See also \cite{AndBrill, Wheeler}.

\subsection{Non-vacuum solutions}
So far we have only discussed the high-frequency limit of vacuum solutions. It is natural to study some analogue of Burnett's question when matter fields are present:
\begin{problem}\label{prob:non.vacuum}
For suitable physical matter models, characterize the high-frequency limits of solutions to the Einstein--matter system.
\end{problem}

There are a few sub-problems that concerning Problem~\ref{prob:non.vacuum}. One may first want to study the high-frequency limits of solutions to the Einstein--massless Vlasov system, as they naturally arise as limits of vacuum solutions.
\begin{problem}
Do high-frequency limits of solutions to the Einstein--massless Vlasov system necessarily solve the Einstein--massless Vlasov system?
\end{problem}
Put differently, is the set of solutions to the Einstein--massless Vlasov system weakly closed? This seems natural to expect if indeed solutions to the Einstein--massless Vlasov system exhaust all possible weak limits of vacuum solutions as suggested by Conjecture~\ref{conj:forward} and Conjecture~\ref{conj:backward}.

As for coupling with other matter fields, perhaps one could distinguish between matter fields which propagate at the speed of light (such as Maxwell field, scalar field, etc.) and those which propagate at a slower speed (such as Euler, massive Vlasov, etc.). 

For matter fields which propagate at the speed of light, one should in principle be able to use the techniques introduced in the works surveyed above to determine the equations for the limiting solutions. At least in the setting of angularly regular spacetimes in Section~\ref{sec:angularly.regular}, the low-regularity existence result in Theorem~\ref{thm:LR2} holds also more generally for the Einstein--Maxwell system or the Einstein--scalar field system with exactly the same proof. This should allow one to extract a limit and to analyze the limiting spacetime. More generally, one can study the following problem:
\begin{problem}
Characterize the high-frequency limits of solutions to the Einstein--Maxwell system or the Einstein--scalar field system.
\end{problem}

The case when the Einstein equations is coupled with Euler matter or massive Vlasov matter may even be more interesting. In fact, even the question concerning weak limits of solutions to the (non-relativistic) Euler equations (without coupling to Einstein) has attracted a lot of interest in connection to turbulence and the Onsager conjecture \cite{DLeSz1, DLeSz2, Isett.book, Isett.Onsager}. 
\begin{problem}
Characterize the high-frequency limits of solutions to the Einstein--Euler system or the Einstein--massive Vlasov system.
\end{problem}

\subsection{Semi-classical limits for the Einstein--Klein Gordon system}

Beyond the limit \eqref{eq:intro.HF.def}, one can study other forms of high-frequency limits. One possible example would be to consider the Einstein--Klein--Gordon system, and consider high-frequency limits simultaneously with the semi-classical limits of the Klein--Gordon equation, i.e.,
$$\Box_g \phi + \hbar^{-2} \phi = 0 \quad \hbox{as $\hbar \to 0$}.$$

\begin{problem} 
Characterize the high-frequency limits of solutions to Einstein--Klein--Gordon equations in the semi-classical regime.
\end{problem}

\bibliographystyle{DLplain}
\bibliography{HFlimit}

\end{document}